%% file: sp-main.tex
\documentclass[compsoc, conference, letterpaper, 10pt, times]{IEEEtran} 
\usepackage{epsfig}
\usepackage{subfigure}
\input{preamble}

\title{On the Economics of Offline Password Cracking} 
\author{Jeremiah Blocki \\ Purdue University \and Ben Harsha \\ Purdue University \and Samson Zhou \\ Carnegie Mellon University}
\begin{document}

\maketitle
\input{abstract}
\input{intro}

\input{preliminaries}

\input{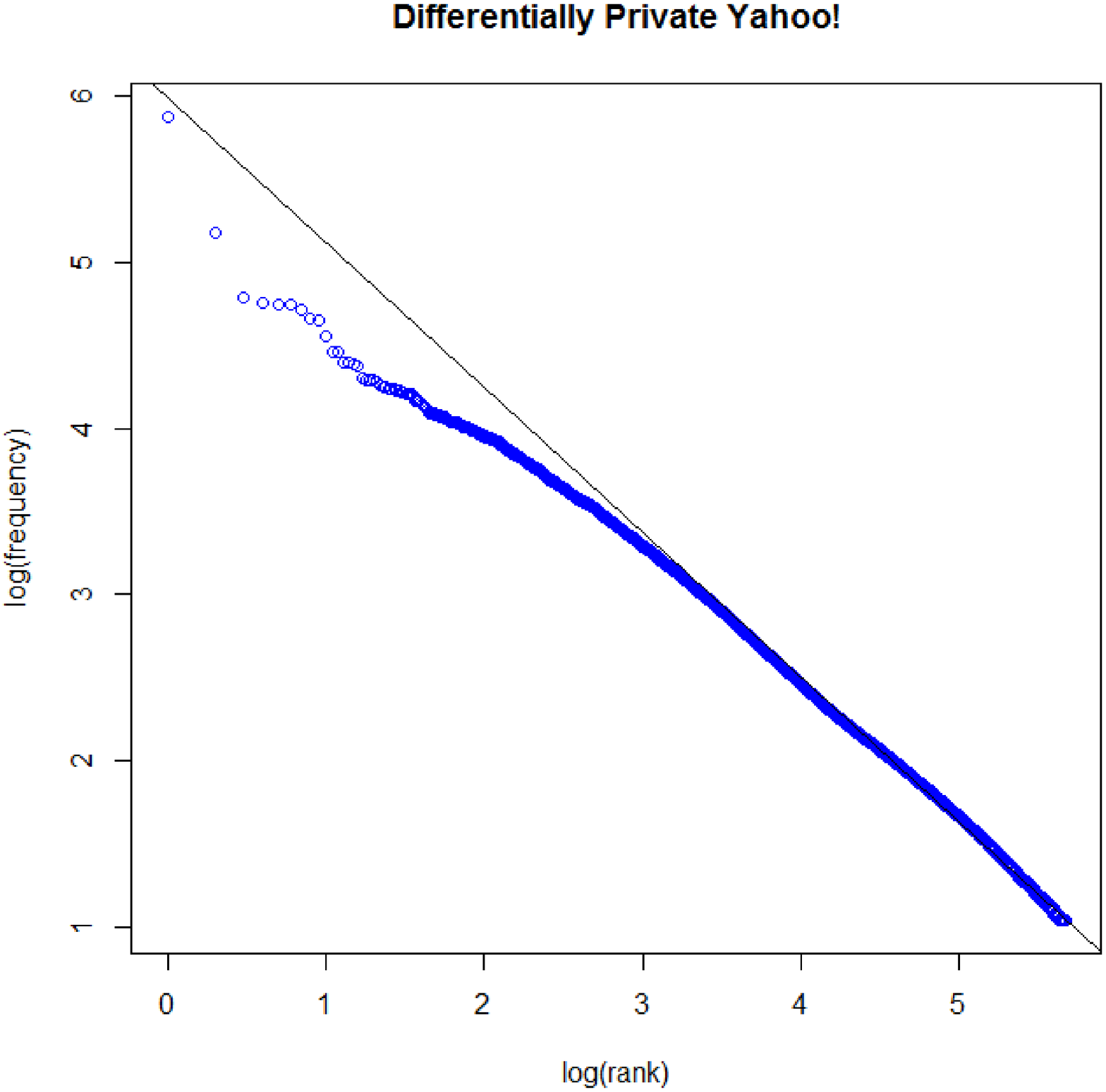}

\input{analysis}
\input{applications}

\input{modelIndependent}
\input{related}
\input{discussion}

\bibliographystyle{IEEEtran}
\bibliography{abbrev0,crypto,jit}
\appendix
\input{missingProofs}
\end{document}

%% file: preamble.tex
\usepackage{amsmath,amssymb,amsfonts}
\usepackage{caption}
\usepackage{euler}
\usepackage{fancybox}
\usepackage{hyperref}
\usepackage{paralist}
\usepackage{color}
\usepackage{wrapfig}
\usepackage{tikz}
\usepackage{setspace}
\usepackage[framemethod=tikz]{mdframed}
\usepackage{xspace}
\usepackage{pgfplots}
\pgfplotsset{compat=1.5}

\newenvironment{proof}{\noindent{\bf Proof : \ }}{\hfill$\Box$\par\medskip}

\newtheorem{theorem}{Theorem}
\newtheorem{corollary}[theorem]{Corollary}

\newtheorem{claim}[theorem]{Claim}

\newenvironment{remindertheorem}[1]{\medskip \noindent {\bf Reminder of  #1.  }\em}{}

\newenvironment{proofof}[1]{\begin{trivlist} \item {\bf Proof
#1:~~}}
  {\qed\end{trivlist}}
\renewenvironment{proofof}[1]{\par\medskip\noindent{\bf Proof of #1: \ }}{\hfill$\Box$\par\medskip}

\newcommand{\COMMENTED}[1]{{}}

\newcommand{\PPr}[1]{\ensuremath{\mathbf{Pr}\Big[#1\Big]}}

\newcommand{\EEx}[1]{\ensuremath{\mathbf{E}\Big[#1\Big]}}

\newcommand{\eps}{\epsilon}

\renewcommand{\figurename}{Fig.}

\newcommand{\ignore}[1]{}

\newif\ifnotes\notesfalse
\ifnotes
\newcommand{\bnote}[1]{\textcolor{red}{{\bf (Ben:} {#1}{\bf ) }} \marginpar{\tiny\bf
             \begin{minipage}[t]{0.5in}
               \raggedright b:
                \end{minipage}}}
\newcommand{\jnote}[1]{\textcolor{blue}{{\bf (Jeremiah:} {#1}{\bf ) }} \marginpar{\tiny\bf
             \begin{minipage}[t]{0.5in}
               \raggedright J:
                \end{minipage}}}

\newcommand{\snote}[1]{\textcolor{green}{{\bf (Samson:} {#1}{\bf ) }} \marginpar{\tiny\bf
             \begin{minipage}[t]{0.5in}
               \raggedright S:
            \end{minipage}}}            							
\else
\newcommand{\bnote}[1]{}
\newcommand{\snote}[1]{}
\newcommand{\jnote}[1]{}
\fi

%% file: abstract.tex
% !TEX root =sp-main.tex
\begin{abstract}
%\subsection*{Abstract}
We develop an economic model of an offline password cracker which allows us to make quantitative predictions about the fraction of accounts that a rational password attacker would crack in the event of an authentication server breach. We apply our economic model to analyze recent massive password breaches at Yahoo!, Dropbox, LastPass and AshleyMadison. All four organizations were using key-stretching to protect user passwords. In fact, LastPass' use of PBKDF2-SHA256 with $10^5$ hash iterations exceeds 2017 NIST minimum recommendation by an order of magnitude. Nevertheless, our analysis paints a bleak picture: the adopted key-stretching levels provide insufficient protection for user passwords. In particular, we present strong evidence that most user passwords follow a Zipf's law distribution, and characterize the behavior of a rational attacker when user passwords are selected from a Zipf's law distribution. We show that there is a finite threshold  which depends on the Zipf's law parameters that characterizes the behavior of a rational attacker ---  if the value of a cracked password (normalized by the cost of computing the password hash function) exceeds this threshold then the adversary's optimal strategy is {\em always} to continue attacking until each user password has been cracked. In all cases (Yahoo!, Dropbox, LastPass and AshleyMadison) we find that the value of a cracked password almost certainly exceeds this threshold meaning that a rational attacker would crack all passwords that are selected from the Zipf's law distribution (i.e., most user passwords). This prediction holds even if we incorporate an aggressive model of diminishing returns for the attacker (e.g., the total value of $500$ million cracked passwords is less than $100$ times the total value of $5$ million passwords). On a positive note our analysis demonstrates that memory hard functions (MHFs) such as SCRYPT or Argon2i can significantly reduce the damage of an offline attack. In particular, we find that because MHFs substantially increase guessing costs a rational attacker will give up well before he cracks most user passwords and this prediction holds even if the attacker does not encounter diminishing returns for additional cracked passwords. Based on our analysis we advocate that password hashing standards should be updated to require the use of memory hard functions for password hashing and disallow the use of non-memory hard functions such as BCRYPT or PBKDF2. 
 %Thus, our model can provide concrete guidance for an organization to follow when tuning the hardness parameters of its password hashing algorithm. When the user password distribution follows Zipf's law we show that there is a finite threshold  which characterizes the behavior of the attacker. In particular, if the value of a cracked password (normalized by the cost of computing the password hash function) exceeds this threshold then the adversary's optimal strategy is {\em always} to continue attacking until each user password has been cracked. We provide strong evidence that the distribution of user passwords does indeed follow Zipf's law by showing that it fits a dataset of $70$ million passwords. Finally, we apply our economic model to analyze recent massive password breaches at Yahoo!, Dropbox, LastPass and AshleyMadison. All four organizations adopted key-stretching tools to protect user passwords. Nevertheless, our analysis paints a bleak picture: the adopted key-stretching levels provide insufficient protection for user passwords. In particular, if the password distribution follows Zipf's law and the adversary's valuation of a cracked password is not {\em significantly} less than blackmarket prices then a rational attacker will eventually crack all user passwords.  On a positive note our analysis suggests that memory hard functions can significantly reduce the damage of an offline attack by forcing a rational attacker to give up much earlier.
\end{abstract}

%% file: intro.tex
% !TEX root =sp-main.tex
\section{Introduction} \label{sec:intro}
In the last few years breaches at organizations like Yahoo!, Dropbox, Lastpass, AshleyMadison, LinkedIn, eBay and Adult FriendFinder have exposed over a billion user passwords to offline attacks. Password hashing algorithms are a critical last line of defense against an offline attacker who has stolen password hash values from an authentication server. An attacker who has stolen a user's password hash value can attempt to crack each user's password offline by comparing the hashes of likely password guesses with the stolen hash value. Because the attacker can check each guess offline it is no longer possible to lockout the adversary after several incorrect guesses. 

An offline attacker is limited only by the cost of computing the hash function. Ideally, the password hashing algorithm should be moderately expensive to compute so that it is prohibitively expensive for an offline attacker to crack most user passwords e.g., by checking millions, billions or even trillions of password guesses for each user. It is perhaps encouraging that AshleyMadison, Dropbox, LastPass and Yahoo! had adopted slow password hashing algorithms like BCRYPT and PBKDF2-SHA256 to discourage an offline attacker from cracking passwords. In the aftermath of these breaches, the claim that slow password hashing algorithms like BCRYPT~\cite{provos1999bcrypt} or PBKDF2~\cite{kaliski2000pkcs} are sufficient to protect most user passwords from offline attackers has been repeated frequently. For example, LastPass~\cite{LastPassBreach} claimed that ``Cracking our algorithms [PBKDF2-SHA256] is extremely difficult, even for the strongest of computers.'' Security experts have made similar claims about BCRYPT e.g., after the Dropbox breach~\cite{DropboxBreach} a prominent security expert confidently stated that ``all but the worst possible password choices are going to remain secure'' because Dropbox had used the BCRYPT hashing algorithm. 

Are these strong claims about the security of BCRYPT and PBKDF2 true? Despite all of their problems passwords remain prevalent and are likely to remain entrenched as the dominant form of authentication on the internet for years to come because they are easy to use and deploy, and users are already familiar with them~\cite{bonneau2012quest,Herley2012,bonneau2015passwords}. It is therefore imperative to develop tools to quantify the damages of password breaches, and provide guidance to organizations on how to store passwords. In this work we seek to address the following question:

\begin{quote}
Can we quantitatively predict how many user passwords a rational attacker will crack after a breach?
\end{quote}

We introduce a game-theoretic model to answer this question and analyze recent data-breaches. Our analysis strongly challenges the claim that BCRYPT and PBKDF2-SHA256 provide adequate protection for user passwords. On the positive side our analysis indicates that more modern password hashing algorithms~\cite{PHC} (e.g., memory hard functions~\cite{Per09}) can provide meaningful protection against offline attackers. 

\subsection{Contributions} We first develop a new decision-theoretic framework to quantify the damage of an offline attack. Our model generalizes the stackelberg game-theoretic model of Blocki and Datta~\cite{BlockiD16}. A rational password attacker is economically motivated and will quit guessing once his marginal guessing costs exceed his marginal reward. The attacker's marginal reward is given by the probability $p_i$ that the next ($i$th) password guess is correct times the value of an {\em additional} cracked password to the adversary e.g., the additional revenue of selling that password on the black market or the expected amount of additional money that could be extorted from this user. Given the average value $v$ of each cracked password for the adversary\footnote{More precisely, if there are $N$ users in the dataset and the total value of all $N$ cracked passwords is $V$ then $v= V/N$. When there are diminishing returns for additional cracked passwords the parameter $v$ may be significantly lower than the value of the first cracked password. }, the cost $k$ of computing the password hash function and the probability distribution $p_1 > p_2  > \ldots$ over user selected passwords, our model allows us to predict exactly how many passwords a rational adversary will crack. Unlike the model of Blocki and Datta~\cite{BlockiD16} we can use our framework to model a setting in which the attacker encounters diminishing returns as we would expect in most (black)markets i.e., the total value of $500$ million cracked passwords may be significantly less than $100$ times the total value of $5$ million passwords.  

Second, we present the strongest evidence to date that Zipf's law models the distribution of user selected passwords (with the possible exception of the tail of the distribution). These findings strongly support previous conclusions of Wang and Wang~\cite{WangW16}. In particular, we show that Zipf's law closely fits the Yahoo! password frequency corpus. This dataset was collected by Bonneau~\cite{bonneau2012yahoo} and later published by Blocki et al.~\cite{DPPLists}. In contrast to datasets from password breaches the Yahoo! dataset was collected by trusted parties, and is representative of active Yahoo! users (researchers have observed that hacked datasets contain many passwords that appear to be fake~\cite{yang2016comparing}). Our sample size, $70$ million users, is also more than twice as large as the datasets Wang and Wang\cite{WangW16} used to support their argument that Zipf's law closely models password datasets. 

Third, we show that there is a finite threshold $T(.)$ which characterizes the behavior of a rational value $v$-adversary whenever the distribution over passwords follows Zipf's law. In particular, if the first cracked password has value $v\geq T(.) \times k$ then the adversary's optimal strategy is always to continue guessing until he cracks the user's password. The threshold $T(y,r,a)$ is parameterized Zipf's law parameters $y$ and $r$ and a parameter $a$ representing the rate of password value decay. We remark that, even if Zipf's law fails to model the tail of the password distribution, the threshold $T(y,r,a)$ still provides a useful characterization of the attacker's behavior. In particular, if $(1-x)\%$ of passwords in a distribution follow Zip's law and the other $x\%$ follow some unknown (possibly uncrackable) distribution then our bounds imply that an attacker will compromise at least $(1-x)\%$ of user passwords whenever $v\geq T(y,r,a) \times k$.

Fourth, we also derive model independent upper and lower bounds on the fraction of passwords that a rational adversary would crack. While these bounds are slightly weaker than the bounds we can derive using Zipf's law these bounds do not require any modeling assumptions e.g., it is impossible to determine for sure whether or not Zipf's law fits the tail of the password distribution. Interestingly, the lower bounds we derive suggest that state of the art password crackers~\cite{melicher2016fast} could still be improved substantially.

Fifth, we apply our framework to analyze recent large scale password breaches including LastPass, AshleyMadison, Dropbox and Yahoo! Our analysis strongly challenges the claim that BCRYPT and PBKDF2-SHA256 provide adequate protection for user passwords. In fact, if the password distribution follows Zipf's law then our analysis indicates that a rational attacker will almost certainly crack 100\% of user passwords e.g.,  unless the value of Dropbox/LastPass/AshleyMadison/Yahoo! passwords is {\em significantly} less valuable than black market projections~\cite{passwordBlackMarket}. 

Finally, we derive {\em model independent} upper and lower bounds on the $\%$ of passwords cracked by a rational adversary. These bounds do not rely on the assumption that Zipf's law models the tail of the password distribution\footnote{Wang and Wang~\cite{WangW16} observed that the tails of empirical password datasets are not inconsistent with a Zipf's law distribution. However, we cannot be entirely confident that Zipf's law models the tail of the distribution since, by definition, we do not have many samples for passwords in the tail of the distribution.  }. Nevertheless, our predictions are still quite dire e.g., a rational adversary will crack $51\%$ of Yahoo! passwords {\em at minimum}. Our analysis indicates that, to achieve sufficient levels of protection with BCRYPT or PBKDF2, it would be {\em necessary} to run these algorithms for well over a second on modern CPU which would constitute an unacceptable authentication delay in many contexts \cite{millerDelay}. On a more positive note our analysis suggests that the use of more modern password hashing techniques like memory hard functions {\em can} provide strong protection against a rational password attacker {\em without} introducing inordinate delays for users during authentication. In particular, our analysis suggests that it could be possible to reduce the $\%$ of cracked passwords below $22.2\%$ without increasing authentication delays to a full second.  

\subsection{Discussion} 
In light of our analysis we contend that that there is a clear need to update standards for password storage to provide developers with clear guidance about the importance of using memory hard functions such as SCRYPT~\cite{Per09} or Argon2id~\cite{Argon2}. In a recent recent user study Naiakshina et al.~\cite{PasswordStorageUserStudy} asked developers to select a password hash function for a new social networking platform. None of the developers in this study selected a memory hard function\footnote{On a positive note the authors did find that priming developers about the importance password security resulted in the selection of stronger password hashing algorithms.} and the strongest password hashing algorithms selected were PBKDF2 with 20,000 hash iterations and BCRYPT with 1,024 iterations. The selection of PBKDF2 with 20,000 hash iterations would be deemed acceptable under 2017 NIST standards~\cite{NIST2017digital} --- PBKDF2 with at least $10,000$ iterations is presented an acceptable selection for password hashing\footnote{An upgrade from $1,000$ iterations as the minimal acceptable number of hash iterations for PBKDF2 in an older 2010 NIST standard~\cite{NISTold}.}. In this sense, LastPass' use of PBKDF2-SHA256 with $100,000$ iterations greatly exceeds current NIST standards.  Nevertheless, our analysis suggests that even PBKDF2-SHA256 with $100,000$ hash iterations is insufficient to protect a majority a user passwords while memory hard functions such as SCRYPT~\cite{Per09} or Argon2id~\cite{Argon2} would provide meaningful protection.  In addition to memory hard functions we also advocate for the use of secure distributed password hashing protocols~\cite{camenisch2012practical,everspaugh2015pythia,phoenix} whenever feasible so that an attacker cannot mount an offline attack without breaching {\em multiple} authentication servers.

%% file: preliminaries.tex
% !TEX root = sp-main.tex
\section{Economic Model} \label{sec:prelim}

\subsection{Preliminaries}
Given a dataset $D$ of $N$ user passwords we use $f_i$ to denote the frequency of the $i$'th most common password in the dataset and we use $pwd_i$ to denote the $i$'th most common password in the dataset. We use $p_1,p_2, \ldots$ to denote the actual distribution over passwords $pwd_1,pwd_2,\ldots$. That is $p_i$ is the probability that a random user selects password $pwd_i$. We use $\hat{p}_i = f_i/N$ to denote an empirical estimate of $p_i$ given a dataset $D$ which was sampled from the real password distribution. We also use $\lambda_i = \sum_{j=1}^i p_j$ to denote the cumulative probability of the $i$ most likely passwords. Equivalently, $\lambda_i$ denotes the probability that an adversary cracks the user's password within the first $i$ guesses. 

We say that the probability distribution $p_1 \geq p_2 \ldots$ follows Zipf's law if $p_i = \frac{z}{i^s}$ for some constants $s$ and $z$.  We say that a probability distribution follows a CDF-Zipf distribution if $\lambda_i = yi^{r}$ for some constants $r$ and $y$. 

\paragraph{Offline Attack} To authenticate users password authentication servers traditionally store salted password hashes. In more detail to authenticate user $u$ the authentication stores a record like the following: $\left(u,s_u,H\left(pwd_u|s_u\right)\right)$. Here, $u$ is the the username and $pwd_u$ is the user's password, $s_u$ is a random string called the salt value used to protect against rainbow table attacks and $H$ is a cryptographic hash function. An adversary who breaches the authentication server will be able to obtain the hash value along with the secret salt value. This adversary can now attempt as many guesses as he desires offline by computing the hashes of likely passwords guesses $H(g_1,s_u), H(g_2,s_u),\ldots$ and comparing these values with the stolen password hash. The attacker is only limited by the resources that he is willing to invest trying to crack the user's password. 

\paragraph{Key Questions and Parameters} We aim to address the following questions: How many guesses will our rational adversary attempt? What fraction of the user passwords will an adversary manage to break? The answer to these questions will depend on several factors. How valuable is a cracked password to the adversary? How much does it cost to compute $H$ each time we validate a new password guess? What does the distribution over user passwords look like?

We use $v$ to denote the value of a cracked password to the adversary measured in units of $C_H$, where $H$ is an underlyng cryptographic hash function like SHA256. We can estimate $v^\$$ by looking at black market prices for cracked passwords. For example, Fossi et al. \cite{passwordBlackMarket} found that the market price for hacked passwords tends to lie in the range $[\$4,\$30]$. A more recent analysis of Yahoo! passwords found that they sell for between $\$0.7$ and $\$1.2$~\cite{stockley_2016} --- the drop in price may be due to an increased supply of Yahoo! passwords. Herley and Florencio found that dishonest behavior can significantly inhibit trade on black markets~\cite{goldForSilver}. Thus, these prices may underestimate the true value of a cracked password. 

Password hash functions are often constructed from an underlying cryptographic hash function $H$. For example, PBKDF2-SHA256 simply iterates the SHA256 hash function multiple times. We use $k$ to denote the cost of a computing the final password hash function --- once again measured in units of $C_H$.  We use $v^\$ = v \times C_H$ (resp. $k^\$ = k \times C_H$) to denote the value (resp. cost) in USD given an estimate of $C_H$.

\subsection{Rational Adversary}
We model a rational adversary who has obtained the salted password hash of a user's password. Our model generalizes the stackelberg game-theoretic framework of Blocki and Datta~\cite{BlockiD16} by introducing a parameter $0 \leq a \leq 1$ which models diminishing returns. We assume that adversary knows the password distribution $p_1,p_2,\ldots$ as well as the corresponding passwords $pwd_1,pwd_2,\ldots$. However, the adversary does not know which password the user selected. 

\paragraph{Attacker Game}We model password cracking using a single-shot game. In the game we sample a random password $pwd$ from the password distribution $\Pr[pwd_i]=p_i$. The adversary picks a threshold $t \geq 0$. The threshold $t$ specifies an ordered list $L(t) = pwd_1,\ldots,pwd_t$ of the $t$ most likely passwords. If the real password is contained in the list of adversary guesses, $pwd \in L(t)$, then the adversary receives a payment of $v$ and we charge the adversary $j \cdot k$, where $j$ is the index of the correct password guess $pwd = pwd_j$. If the real password is not contained in the list $pwd_1,\ldots,pwd_t$ of adversary guesses then the adversary receives no payment ($v=0$) and the adversary is charged $t \cdot k$. Notice that $t=0$ corresponds to the strategy in which the adversary gives up without guessing, and $t=\infty$ corresponds to the strategy in which the adversary never quits. Observe that $\lambda_t = \sum_{j=1}^t p_j$  denotes the fraction of user passwords that are cracked by a threshold $t$ adversary.

\paragraph{About the Attacker} In our analysis we consider an attacker that is 
\begin{enumerate}
\item {\bf Informed:} The attacker knows the password distribution $p_1,p_2\ldots$ and the associated passwords $pwd_1,pwd_2\ldots$. However, the attacker does not know which password a particular user $u$ selected.
\item {\bf  Untargeted:} The attacker does not have personal knowledge about the user that can be exploited to improve the guessing attack. 
\item {\bf Rational:} The attacker is economically motivated, and will stop attacking the user once marginal guessing costs exceed the marginal guessing rewards.
\end{enumerate}

\paragraph{Discussion} Our attacker model captures the most common types of password attacks. It is generally reasonable to assume that the attacker knows the password distribution --- possibly excluding of the tail of the distribution. In particular, previous password breaches provide plenty of training data for the attacker and it is reasonable to assume that password cracking models will continue to improve as attackers obtain more and more training data from future password breaches.  We focus on an untargeted attacker in our analysis. However, we stress that our model may also be useful when considering a targeted attacker with background knowledge of the user (e.g., name, birthdate, hobbies etc...). In particular, let $p_i$ denote the probability that a targeted adversary's $i$'th guess is correct. Wang et al. observed that a targeted distribution over user passwords $p_1,p_2 \ldots$ still seems to follow Zipf's law~\cite{wang2016targeted}. 

\paragraph{Rational Attacker Behavior} If the adversary chooses a threshold $t$ then his expected guessing costs are 

\[C(t)= t\left(1- \sum_{j=1}^t p_j \right) k + k\sum_{j=1}^t j\cdot p_j  \ . \]
Similarly, his expected reward is 
\[ R(t) = v \left(\sum_{j=1}^t p_j\right)^a \ \]
where the parameter $0 \leq a \leq 1$ allows us to model diminishing returns for the attacker as he obtains additional cracked passwords.  For example, let $t_{1\%}$ (resp. $t_{2\%}$) be  given such that $p_1+\ldots+p_{t_{1\%}}=0.01$ ($p_1+\ldots+p_{t_{2\%}}=0.02$) then for $a<1$ we have $R(t_{2\%}) = 2^a R(t_{1\%}) < 2 \times R(t_{1\%})$ even though an adversary cracks twice as many passwords by increasing his threshold from $t_{1\%}$ to $t_{2\%}$.  

\noindent {\bf Diminishing Returns: } We note that the original model of Blocki and Datta~\cite{BlockiD16} is a special case of our model when $a=1$ (no diminishing returns). There are a number of reasons why an attacker may encounter diminishing returns $(a<1)$ for additional cracked passwords. First, if the attacker plans to sell the passwords on the black market then basic economics suggests that increasing the supply of cracked passwords is likely to drive down prices. In the case of a large breach like Yahoo! (500 million passwords) it is conceivable the number of available passwords on the black market might quickly increase by two orders of magnitude. Second, the more user accounts that are hacked/actively exploited the more likely it is that the original breach will be detected. If the breach is detected then an organization can ask (or require) users to change their passwords or require two-factor authentication, which will reduce the value of each cracked password\footnote{However, the cracked passwords arguably still have significant value after the breach is detected for two reasons. First, many users will not update their passwords unless they are required to do so. Second, many of the users that do update their passwords may do so in a predictable way~\cite{CCS:ZhaMonRei10}. Third, many users will have the same password for other accounts.}. 

\noindent {\bf Interpreting model parameter $v$: } We note that we have $v = R(\infty)$, where $R(\infty) \times N$ denotes the total value of a completely cracked password dataset of size $N$. Thus, the parameter $v$ denotes the average value of a cracked password given that all password have been cracked. We can estimate this parameter $v$ based on black market sales data. For example, suppose that we know that $R(t_{1\%}) = \$4 \times 1\%$ e.g., from equilibrium black market prices when only $1\%$ of cracked passwords are on the market. In this case we can extrapolate 
\begin{equation} v = R(\infty) = R(t_{100\%}) = 100^a R(t_{1\%}) = \frac{\$4}{100^{1-a}} \ . \label{eq:valueFormula} \end{equation}

\noindent{\bf Rational Attacker Behavior:} Formally, the rational adversary will select the threshold $t^*$ maximizing his overall utility \[t^* = \arg\max_{t} \left( R(t)-C(t) \right) \ . \] 

Intuitively, a rational adversary should stop guessing if the marginal cost of one more password guess exceeds the marginal benefit of that guess. Thus, we will have $MC(t^*)=C(t^*)-C(t^*-1) \approx MR(t^*) = R(t^*)-R(t^*-1)$. The marginal cost of increasing the threshold from $t-1$ to $t$ is  \begin{equation} MC(t) = C(t)-C(t-1) = k\left(1-\sum_{j=1}^{t-1} p_j\right) \ . \label{eq:mc} \end{equation}
Intuitively, the attacker pays an extra cost $k$ to hash $pwd_t$ if and only if the first $t-1$ guesses are incorrect. Similarly, the attacker's marginal revenue is $MR(t) = R(t)-R(t-1)$ when $a=1$ we have $MR(t)=v \times p_t$ otherwise \begin{equation}MR(t) = v\left( \left(\sum_{j=1}^t p_j\right)^a-\left(\sum_{j=1}^{t-1} p_j\right)^a \right)\times p_t \ . \label{eq:MR}\end{equation}
Note that $\lambda_{t^*}$ denotes the expected fraction of passwords compromised by an rational attacker.  Given a specific assumption about the password distribution (e.g., Zipf's law) we can derive bounds on $\lambda_{t^*}$. 
 
\paragraph{Competition} We do not attempt to directly model the behavior of an adversary who faces competition from other password crackers. Many breaches (e.g., Yahoo!, LinkedIn, Dropbox) remained undetected for several years. In these cases it may be reasonable to assume that the password cracker faced no competition. However, competition certainly could occur in the event that the breach is public (e.g., Ashley Madison). In an extremely competitive setting (e.g., password for a cryptocurrency wallet) only the first attacker to crack the password will be rewarded\footnote{However, we remark that in many instances attackers may unknowingly ``share'' the benefit of a cracked account. For example, an attacker who cracks a password may not actually change the password since such an action would alert the legitimate user of the breach.}. Such competition would decrease the expected reward for each cracked password and could potential reduce the total $\%$ of passwords cracked by {\em each individual} attacker. 

However, from the defender's point of view the goal is to minimize the $\%$ of passwords that are cracked by {\em any} attacker. Thus, we can argue that competition will have a minimal impact on the total $\%$ of cracked passwords. In particular, even in an extremely competitive setting where only the first attacker to find the password is rewarded we still have 
\[ \mathbf{CompCrack}(v,a) \geq \min_{0 \leq p \leq 1} \max \{ \mathbf{Cracked}(pv,a), 1-p\} \ . \]
Here $\mathbf{CompCrack}(v)$ (resp. $\mathbf{Cracked}(v)$) denotes the $\%$ of passwords that are cracked by some attacker when the value of a password is $v$ and attackers face competition (resp. do not face competition). This follows because the expected reward for attacker when faced with competition is at least $R_{comp}(t) \geq p_{first} \times R(t)$ where $p_{first}$ is the probability that no competing attacker managed to crack the password already. If $p_{first}$ is small then the marginal rewards will also be small so the attacker may quit earlier, but in this case it is likely  that another attacker has already compromised the account ($1-p_{first}$).  
 
 \paragraph{Defender Actions}
 The value $\lambda_{t^*}$ will depend on $k$, $v$ as well as the underlying password distribution $p_1 \geq p_2 \geq \ldots$. The goal of key-stretching is to increase $k$ so that we can reduce $\lambda_{t^*}$, the fraction of compromised accounts, in the event of an authentication server breach. However, the defender is constrained by server workload and by authentication times. In particular, the number of sequential hash iterations ($\tau$) is bounded by usability constraints as users may be unhappy if they need to wait a long time to authenticate e.g., it would at least a second to compute PBKDF2-SHA256 with $\tau=10^7$ hash iterations on a modern CPU~\cite{BonneauS14}. Similarly, the total workload $k$ is similarly bounded by workload constraints e.g., the authentication server must be able to handle all of the authentication requests even during traffic peaks. If the value $v$ is sufficiently large (in proportion to the cost $k$ of a password guess) then a rational attacker will crack every password $\lambda_{t^*} = 1$. In this case we say that all of the key-stretching effort was useless against a value $v$ rational adversary.
 
  Password hashing algorithms like BCRYPT, PBKDF2 and SCRYPT have parameters that control the running time (number of hash iterations) $\tau$ and total cost $k$ of computing the password hash function. Thus, the cost $k$ of computing PBKDF2 or BCRYPT is $k=\tau \times C_H$, where $C_H$ denotes the cost of computing the underlying hash function (e.g., SHA256 or Blowfish). We will treat $C_H$ as a unit of measurement when we report the cost $k$ and write $k=\tau$ for the BCRYPT and PBKDF2 functions. Given an estimate of $C_H$ in USD we will use $k^\$ = k \times C_H$ to denote the cost of computing the password hash function in USD. 
  
Intuitively, a memory hard function is a function whose computation requires large amounts of memory. One of the key advantages of a memory hard function is that cost $k$ potentially scales with $\tau^2$ instead of $\tau$ making it possible to increase costs without introducing intolerable authentication delays.  An ideal memory hard function runs in time $\tau$ and requires $\tau$ blocks of memory to compute. Thus, the Area x Time (AT) complexity of computing the Memory Hard Function scales with $\tau^2$ because the adversary must allocate $\tau$ blocks of memory for $\tau$ units of time. In particular, we use $k=\tau \times C_H+\tau^2 \times C_{mem}$ to model the approximate cost of computing a memory hard function which iteratively makes $\tau$ calls to the underlying hash function $H$ and requires $\tau$ blocks of memory. By contrast, the AT complexity of BCRYPT and PBKDF2 is just $k = \tau$ since these functions can be computed with a single block of memory. Here, $C_{mem}$ is a constant representing the core memory-area ratio. That is the area of one block of memory on chip divided by the the area of a core evaluating $H$ on chip. In this paper we use the estimate $C_{mem} \approx 1/3000$ as in ~\cite{BK15,AB16} though we stress that our analysis could be easily repeated with different parameter choices.

\paragraph{Model Limitations} To keep exposition simple we do not attempt to incorporate any model of equilibrium prices for cracked passwords on the black market and instead assume that the value of a cracked password $v^\$$ is static for all users. A targeted adversary may have higher valuations for specific user passwords e.g., celebrities, politicians. Similarly, an attacker who floods a black market with cracked passwords may drive equilibrium prices down. Our primary findings would not be altered in any significant way by including such a model unless equilibrium prices drop by $1$--$2$ orders of magnitude~\cite{passwordBlackMarket}. We also remark that our intention is to model an untargeted economically motivated attacker and not a nation state focused on cracking the passwords of a particular person of interest. However, it may still be reasonable to believe that a nation state attacker will be largely be constrained by economic considerations (e.g., expected value of additional intelligence gained by cracking the password versus expected cost to crack password).

%% file: yahoo.tex
%  !TEX root=sp-main.tex
\section{Yahoo! Passwords follow Zipf's Law} \label{sec:ZipfLawYahoo}

\graphicspath{ {images/} }
 Zipf's law states that the frequency of an element in a distribution is related to its rank in the distribution. There are two variants of Zipf's law for passwords: PDF-Zipf and CDF-Zipf. In the CDF-Zipf model we have $\lambda_t = \sum_{j=1}^t p_i = y \cdot t^r$, where the constants $y$ and $r$ are the CDF-Zipf parameters. In the PDF-Zipf model we have  $f_i = \frac{C}{i^s}$, where $s$ and $C$ are the PDF-Zipf parameters. Normalizing by $N$ the number of users we have $p_i = \frac{z}{i^s}$, where $z = \frac{C}{N}$.  

Wang et al.~\cite{WangZipfLaw14} previously found that password frequencies tend to follow PDF-Zipf's law if the tail of the password distribution (e.g., passwords with frequency $f_i < 5$) is dropped. Wang and Wang~\cite{WangW16} subsequently found that CDF-Zipf's model is superior in that the CDF-Zipf fits were more stable than PDF-Zipf fits and that the CDF-Zipf fit performed better under Kolmogorov-Smirnov (KS) tests.  Furthermore, the CDF-Zipf model can fit the entire password distribution (e.g., without excluding passwords with frequency $f_i < 5$). These claims were based on analysis of several smaller password datasets ($N \leq 32.6$ million users) which were released by hackers. 

In 2016 Yahoo! allowed the release of a differentially private list of password frequencies for users of their services~\cite{DPPLists}. We refer an interested reader to ~\cite{bonneau2012yahoo,DPPLists} for additional details about how the Yahoo! data was collected and how it was perturbed to preserve differential privacy. The Yahoo! dataset is superior to other datasets in that it offers the largest sample size $N=70$ million and the dataset was collected and released by trusted parties. We show that the Yahoo! dataset is also well modeled by CDF-Zipf's law. Our analysis comprises the strongest evidence to date of Wang and Wang's premise~\cite{WangW16} that password distributions follow CDF-Zipf's law due to the advantages of the Yahoo! dataset. We focus on the CDF-Zipf's law model in this section since it can fit the entire password distribution~\cite{WangW16}. We also verified that the Yahoo! dataset is also well modeled by PDF-Zipf's law if we drop passwords with frequency $f_i < 5$ like Wang et al.~\cite{WangZipfLaw14}, but we omit this analysis from the submission due to lack of space.

The rest of this section is structured as follows: First, in section \ref{subsec:EcologicalValidity} we discuss the advantages of using the Yahoo! dataset over leaked datasets like RockYou. In \ref{subsec:CDFZipfDP} we show that the noise that was added to preserve differential privacy will have a negligibly small impact on CDF-Zipf fittings. In section \ref{subsec:CDFZipfStability} we use subsampling to show that the CDF-Zipf fittings for Yahoo! converge to a stable solution. Finally, in section \ref{subsec:CDFFitYahoo} we present the CDF-Zipf fitting for the entire Yahoo! dataset.

%Wang et al.~\cite{WangZipfLaw14} and Wang and Wang~\cite{WangW16} previously found that password frequencies tend to follow Zipf's law, basing this claim on analysis of several smaller password datasets ($N \leq 32.6$ million users) released by hackers. 
\subsection{On Ecological Validity} \label{subsec:EcologicalValidity}
The Yahoo! frequency corpus offers many advantages over breached password datasets such as RockYou or Tianya.

\begin{itemize}
	\item The Yahoo! password frequency corpus is based on $70$ million Yahoo! passwords --- more than twice as large as any of the breached datasets analyzed by Wang and Wang~\cite{WangW16}.
	\item The records were collected in a trusted fashion. No infiltration, hacking, tricks, or general foul play was used to obtain any of this data. There was no ulterior motive behind collecting these passwords other than to provide valuable data in a way that can be used for scientific research. By contrast, it is possible that hackers strategically omit (or inject) password data before they release a breached dataset like RockYou or Tianya! Why should we trust rogue hackers to provide researchers with representative password data?
	\item Breached password datasets often contain many passwords/ accounts that look suspiciously fake. In 2016 Yang et al \cite{yang2016comparing} suggested that such passwords can be removed with DBSCAN \cite{ester1996density}. Cleansing operations ended up removing a reasonable portion of the dataset (e.g., 5 million passwords were removed from RockYou's data). With the Yahoo! data such cleansing is not needed, as it was collected in a manner that ensured collected passwords were in use. Previous work that has been done on Zipf distributions in breached password datasets~\cite{WangW16} did not perform any sort of sanitizing step on the data. It is unclear how such operations would affect the Zipf law fit. 
	\item The information is released in a responsible way that preserves users' privacy. The differential privacy mechanism means that even with the released data it is not possible to determine any new information about Yahoo's users that an adversary would not be able to obtain anyways.
	\item Data from the Yahoo! password frequency corpus ultimately is derived from the passwords of active Yahoo! users who were logging in during the course of the study as opposed to passwords from throwaway accounts that have been long forgotten. 

\end{itemize}

\subsection{On the Impact of Differential Privacy on CDF-Zipf Fits} \label{subsec:CDFZipfDP}

\begin{table}
\centering
\begin{tabular}{ | l | c c |}
	\hline
	List Version 			& $y$ 			&  $\sigma_y$		\\ \hline
	RockYou Standard 		& $0.0288$		& 					\\ \hline
	RockYou Diff. Private 	& $0.0302$		& $1.348 * 10^{-6}$	\\ \hline
							& $r$ 			& $\sigma_{r}$		 \\ \hline
	RockYou Standard 		& $0.2108$	 	& 					\\ \hline
	RockYou Diff. Private 	& $0.2077$		& $2.94 * 10^{-6}$ 	\\ \hline
							& $R^2$			& $\sigma_{R^2}$	\\ \hline
	RockYou Standard		& $0.9687$		&					\\ \hline
	RockYou Diff. Private	& $0.9681$		& $6.50 * 10^{-7}$	\\ \hline
% $0.99771$ & $8.28 * 10^{-7}$ &
	 
	 %$.0239$ & $2.74 * 10^{-6}$ \\
%\hline

\end{tabular}
\centering
\caption{Impact of Differential Privacy on CDF Fit}
\label{tab:DP_CDF}
\end{table}

The published Yahoo! password frequency lists were perturbed to ensure differential privacy. Before attempting to fit this dataset using Zipf's law we seek to answer the following question: Does this noise, however small, affect our CDF-Zipf fitting process in any significant way? We claim that the answer is no, and we offer strong empirical evidence in support of this claim. In particular, we took the RockYou dataset ($N\approx 32.6$ million users) and generated $30$ different perturbed versions of the frequency list by running the $\left( \epsilon, \delta \right)$-differentially private algorithm of Blocki et al.~\cite{DPPLists}. We set $\epsilon=0.25$, the same value that was used to collect the Yahoo! dataset that we analyze. For each of these perturbed frequency lists we compute a CDF-Zipf law fit using linear least squares regression. To apply Linear Least Squares regression we apply logarithms to the  CDF-Zipf equation  $\lambda_t = y \cdot t^r$  to obtain a linear equation $\log \lambda_t = \log y + r \log t$.

Our results, shown in Table \ref{tab:DP_CDF}, strongly suggest that the differential privacy mechanism does not impact the parameters $y$ and $r$ in a CDF-Zipf fitting in any significant way. In particular, the parameters $y$ and $r$ we obtain from fitting the original data with a CDF-Zipf model are virtually indistinguishable from the parameters we obtain by fitting on one of the perturbed datasets. Similarly, differential privacy does not affect the $R^2$ value of the CDF-Zipf fit. Here, $R^2$ measures how well the linear regression models the data ($R^2$ values closer to 1 indicate better fittings). Thus, one can compute CDF-Zipf's law parameters for the Yahoo! data collected by \cite{DPPLists} and \cite{bonneau2012yahoo} without worrying about the impact of the $\left( \epsilon, \delta \right)$-differentially private algorithm used to perturb this dataset. We also verified that the noise added to the Yahoo! dataset will also have a negligible affect on the  parameters $s$ and $z$ in a PDF-Zipf fitting. 

%	To investigate the impact of the noise added by differential privacy to the Yahoo! dataset we run a series of experiments on the well-known rockyou dataset. The rockyou data was run through the privacy mechanism $n=30$ times, converted into a CDF, and run through linear regression to obtain the Zipf parameters. The results of these experiments are shown in table \ref{tab:DP_CDF}. From these results we can see that while the privacy mechanism does affect the results, the parameters are still close to the true values of the original dataset.

\subsection{Testing Stability of CDF-Zipf Fit via Subsampling} \label{subsec:CDFZipfStability}
There are two primary ways to find a CDF-Zipf fit: Golden Section Search (GSS) and Linear Least Squares (LLS). Wang et al.~\cite{WangW16} previously found that CDF-Zipf fits stabilize more quickly with GSS than with LLS. This was particularly important because the largest dataset they tested had size $\approx 3\times 10^7$. In this section we test the stability of LLS by subsampling from the much larger Yahoo! dataset.  In particular, we subsample (without replacement) datasets of size $15$ million, $30$ million, $45$ million and $60$ million and use LLS to compute the CDF-Zipf parameters $y$ and $r$ for each subsampled dataset. Our results are shown in table \ref{tab:CDFStability} graphically in Figure \ref{fig:CDFsubsample}. While the CDF-Zipf fit returned by LLS does take longer to stabilize our results indicate that it does eventually stabilize at larger (sub)sample sizes (e.g., the Yahoo! dataset).  
\begin{figure}
\centering
\includegraphics[scale=0.40]{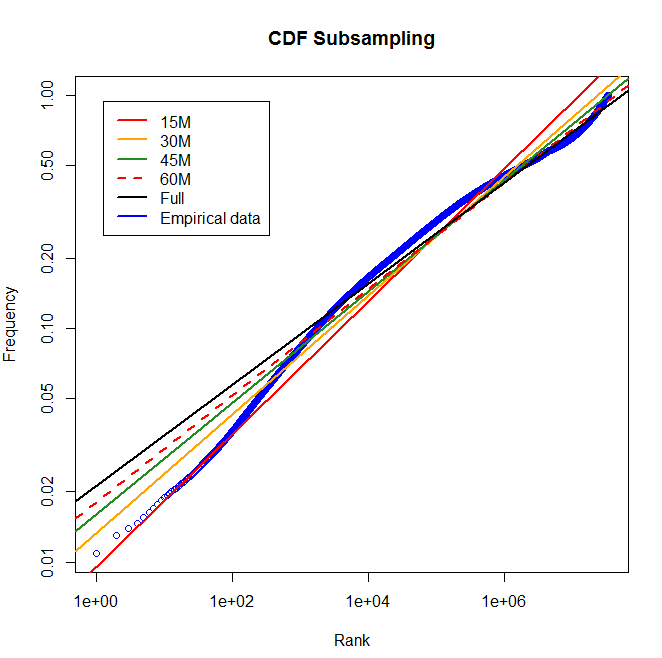}
\caption{Yahoo! CDF-Zipf Subsampling}
\label{fig:CDFsubsample}
\end{figure}

\begin{table}
\centering
\begin{tabular}{| p{0.7in} | c | c | c | }

\hline

Sample Size (Millions)  & $y$  		& $r$     	&  $R^2$   	\\ \hline
15         				& 0.00949  	& 0.2843  	& 0.9542   	\\ \hline
30        				& 0.01321  	& 0.2544 	& 0.9531   	\\ \hline
45        				& 0.01592  	& 0.2384 	& 0.9529 	\\ \hline
60        				& 0.01810  	& 0.2277  	& 0.9530   	\\ \hline
Full       				& 0.02112   & 0.2166  	& 0.9544 	\\ \hline

\end{tabular}
\caption{Yahoo! CDF-Zipf with Sub-sampling}
\label{tab:CDFStability}
\end{table}

\jnote{Mention that there are two ways to find a CDF-Zipf fit: Golden Section search and Linear Least Squares. Wang et al.~\cite{WangW16} previously found that CDF-Zipf fits stabilize quickly with GSS. We find that LLS does stabilize before $N=70$ million examples (albeit slightly more slowly). LLS also has several advantages of GSS in that it can (1) be computed faster, (2) minimizes KS distance, (3) minimizes $R^2$.  }

We also found that the PDF-Zipf parameters $s$ and $z$ stabilize before $N=7\times 10^7$ samples. 
\subsection{Fitting the Yahoo! data set with CDF-Zipf} \label{subsec:CDFFitYahoo}
We used both LLS regression and GSS to obtain separate CDF-Zipf fittings for the Yahoo! dataset. The results, shown in table \ref{tab:YahooCDF_KS} and graphically in Figure \ref{fig:BothFit} showed that both methods produce high quality fittings.  In addition to the parameters $y$ and $r$ we report $R^2$ values and Kolmogorov-Smirnov (KS) distance. The KS test can be thought of as the largest distance between the observed discrete distribution $F_n\left(x\right)$ and the proposed theoretical distribution $F\left(x\right)$. Formally,
$$
D_{KS} = sup\left|F_n(x) - F(x) \right|
$$
Intuitively, smaller $D_{KS}$ values (resp. larger $R^2$ values) indicates better fits. 

%Wang et al.~\cite{WangW16}  recently proposed an improved method of fitting password distributions with Zipf distributions. Rather than using the probability density function for a Zipf distribution they use the cumulative distribution function. In particular, in the CDF-Zipf model we have $\lambda_t = y\cdot i^r$. Wang et al.~\cite{WangW16} previously found that CDF-Zipf fits were more stable  than CDF-Zipf fits and that the fits performed better under Kolmogorov-Smirnov (KS) tests, which measures the maximum distance between fit CDF-Zipf distribution and the empirical distribution. Furthermore, it was not necessary to drop infrequent passwords (e.g., those observed $< 5$ times) to obtain the CDF-Zipf fit. 

%The parameters returned by the GSS, available in table \ref{tab:YahooCDF_KS}, were consistent with data from Wang and Wang's analysis of previous password frequency lists. In addition to running the Golden Section search method, we also ran the LLS regression method used in the PDF Zipf fitting. This was done by using the same procedure of taking the log of the CDF equation. $\lambda_t = y \cdot t^r$ becomes $\log \lambda_t = \log y + r \log t$. 
\ignore{
\begin{table}
\begin{tabular}{|l|c|c|}
\hline
Data set	&	$y$			& 	$r$ 		\\ \hline
Yahoo!		& 	$0.0211$ 	& 	$0.2166$ 	\\ \hline
\end{tabular}
\centering
\caption{Yahoo! CDF-Zipf Parameters}
\label{tab:YahooCDF}
\end{table}
}

\subsubsection{Discussion}
Both LLS and GSS produce high quality CDF-Zipf fittings (e.g., $R^2=0.9544$) for the Yahoo! dataset.  LLS regression outperforms the golden section search under both $R^2$ and Kolmogorov-Smirnov (KS) tests. Wang and Wang~\cite{WangW16} had previously adopted golden section search because the results stabilized quickly. While this was most likely the right choice for smaller password datasets like RockYou, our analysis in the previous section suggest that LLS eventually produces stable solutions when the sample size is large (e.g., $N \geq 60$ million samples) as it is in the Yahoo! dataset. Thus, in the remainder of the paper we use the CDF-Zipf  parameters $y=0.0211$ and $=0.2166$ from LLS regression. We stress that the decision to use the CDF-Zipf parameters from LLS instead of the parameters returned by GSS does not affect our findings in any significant way. 

We  remark that LLS is also more efficient computationally. While we were able to run GSS to find a CDF-Zipf fit for the Yahoo! dataset $(N\approx 7 \times10^7$), running GSS on a dataset of $N=1$ billion passwords (e.g., the size of the most recent Yahoo! breach~\cite{YahooBillionBreach}) would be difficult if not intractable. By contrast, LLS could still be used to find a CDF-Zipf fitting and our analysis suggests that the fit would be superior. 

%Wang and Wang showed experimentally that the results obtained from a Golden Section search are stable when subsampling the available distribution. Our finding was that the CDF fittings converge in the same manner as the PDF converges under subsampling. Taken as a whole, the fitting results suggest that, if possible, it is better to use LSS regression as opposed to golden section. However, with recent breaches such as Yahoo's breach of 1 billion user records  it become increasingly impractical to compute these Zipf parameters using LSS regression. In cases where the size of the dataset becomes large enough that memory becomes an issue, methods that work well with subsamples, such as GSS, are more desirable.

\begin{table}
\vspace{-0.3cm}
\begin{tabular}{|l|c|c|c|c|}
\hline
Method 			&	$y$		&	$r$		&		$R^2$			& 	KS 			\\ \hline
LLS				&	0.0211	&	0.2166	& 		$0.9544$ 		& 	$0.0094328$ 	\\ \hline
GSS	&	0.03315	&	0.1811	& 		$0.9498$ 		& 	$0.022282$ 	\\ \hline
\end{tabular}
\centering

\caption{Yahoo! CDF-Zipf Test Results}
\label{tab:YahooCDF_KS}
\vspace{-0.3cm}
\end{table}

%% file: analysis.tex
% !TEX root=sp-main.tex
\section{Analysis of Rational Adversary Model for Zipf's Law} \label{sec:ZipfLawAnalysis}
In this section, we show that there is a finite threshold $T(y,r,a)$ which characterizes the behavior of a rational offline adversary when user passwords follow CDF-Zipf's law with parameters $y$ and $r$ i.e., $\lambda_i = yi^r$.  In particular, Theorem \ref{thm:threshold} gives a precise formula for computing this threshold $T(y,r,a)$\footnote{We remark that when $a=1$ it is possible to derive a closed form expressing for the threshold $T(y,r,a)$.}. If $v/k \ge T(y,r,a)$ then a rational value $v$ adversary will proceed to crack all user passwords as marginal guessing rewards will {\em always} exceed marginal guessing costs for a rational attacker. In Table \ref{tab:MinThreshold} we use this formula to explicitly compute $T(y,r,a)$ for the Yahoo! dataset as well as for nine other password datasets analyzed by Wang and Wang~\cite{WangW16}.

 We note that we choose to focus on CDF-Zipf's law in this section as it is believed to be better than PDF-Zipf models. However, we stress that similar bounds can be derived using PDF-Zipf's law though we omit these results from the submission for lack of space.  %and we refer an interested reader to appendix \ref{apdx:PDFZipfLawAnalysis}. 

%(1-y (((y/(1-r))^(-1/r))-1)^r)/(y r ((y/(1-r))^(-1/r))^(r-1)) at  r=0.187227 and y=0.037433

\begin{table}
\begin{tabular}{|c|cc|c|c|}\hline
Dataset & $y$ & $r$ & $T(y,r,1)$ & $T(y,r,0.8)$ \\\hline
RockYou & $0.0374$ & $0.1872$ & $1.70\times10^7$ & $2.04\times10^7$ \\\hline
000webhost & $0.0059$ & $0.2816$ & $3.67\times10^7$ & $4.27\times 10^7$  \\\hline
Battlefield & $0.0103$ & $0.2949$ & $2.37\times10^6$  & $2.77 \times 10^6$\\\hline
Tianya & $0.0622$ & $0.1555$ & $2.28\times10^7$ & $2.76\times10^7$  \\\hline
Dodonew & $0.0194$ & $0.2119$ & $4.92\times10^7$ & $5.87\times 10^7$  \\\hline
CSDN & $0.0588$ & $0.1486$ & $7.63\times10^7$ & $9.24\times 10^7$ \\\hline
Mail.ru & $0.0252$ & $0.2182$ & $8.75\times10^6$ &$1.04 \times 10^7$\\\hline
Gmail & $0.0210$ & $0.2257$ & $1.14\times10^7$ & $1.36\times 10^7$\\\hline
Flirtlife.de & $0.0346$ & $0.2916$ & $4.44\times10^4$ & $5.19\times 10^4$   \\\hline
\textcolor{blue}{Yahoo!} & \textcolor{blue}{0.0211} & \textcolor{blue}{0.2166} & \textcolor{blue}{$2.25\times 10^7$} & \textcolor{blue}{$2.69\times10^7$} \\\hline
\end{tabular}

\caption{CDF-Zipf threshold $T(y,r,a)<v/k$ at which adversary cracks $100\%$ of passwords for $a\in \{1,0.8\}$.}
\label{tab:MinThreshold}
\vspace{-0.2cm} 
%\caption{CDF-Zipf threshold $T(y,r)<v/k$ at which adversary cracks $100\%$ of passwords}
%{(1-y*(((1-r)/y)^(1/r)-1)^r)/(y*r*((1-r)/y)^(1-1/r))} for y=0.033150 and r=0.181059
\end{table}

\begin{theorem} \label{thm:threshold}
Let $k$ denote the cost of attempting a password guess. If 
\[ \frac{v}{k}\ge T(y,r,a) = \max_{t \leq Z} \left( \frac{1-y(t-1)^r}{y^a (r a)t^{r a - 1}} \right) 
\]
where 
\[
Z = \left\lceil \left( \frac{1}{y} \right)^{1/r} \right\rceil + 1
\]
then a value $v$ rational attacker will crack $100\%$ of passwords chosen from a Zipf's law distribution with parameters $y$ and $s$.
\end{theorem}
\begin{proof}
Suppose a password frequency distribution follows Zipf's Law, for some parameters $0<r<1$ and $y$, so that $\lambda_n=yn^r$.
% and $p_n=\lambda_n^a-\lambda_{n-1}^a$. 
Since the marginal revenue is $MR(t)=v(\lambda_t^a-\lambda_{t-1}^a)$ and the marginal cost is $MC(t)=k\left(1-\sum_{n=1}^t p_n\right)$, a rational adversary can be assumed to continue attacking as long as $MR(t) \geq MC(t)$. 
Therefore, the attacker will not quit as long as
\begin{align*}
v(\lambda_t^a-\lambda_{t-1}^a) &\geq k\left(1-\sum_{n=1}^t p_n\right) \\
v(y^at^{ra}-y^a(t-1)^{ra}) &\geq k(1-y(t-1)^r)
\end{align*}
In particular, the attacker will not quit as long as
\begin{align*}
\frac{v}{k} &\geq \frac{1-y(t-1)^r}{y^at^{ra}-y^a(t-1)^{ra}} \ . 
\end{align*}
Notably, if $\frac{v}{k}\ge\max_t\left(\frac{1-y(t-1)^r}{y^at^{ra}-y^a(t-1)^{ra}}\right)$ for all $t$, then a rational adversary will eventually crack \emph{all} passwords. 
For $g(t):=y^a (ra) t^{ra-1}$, we have $y^at^{ra}-y^a(t-1)^{ra}=\int_{t-1}^t g(x) \, \mathrm{d}x$. 
Since $ra\le 1$, then $g(t)\le g(x)\le g(t-1)$ for all $x\in[t-1,t]$.  
Thus we have $y^a (ra)t^{ra-1} \leq y^at^{ra}-y^a(t-1)^{ra} \leq y^a (ra)(t-1)^{ra-1}$ and
%Since $y^at^{ra}-y^a(t-1)^{ra}= \int_{t-1}^t y^a (ra) x^{ra-1} \, \mathrm{d}x$, then for $ra\le 1$, we have $y^a (ra)t^{ra-1} \leq y^at^{ra}-y^a(t-1)^{ra} \leq y^a (ra)(t-1)^{ra-1}$ and 
\[\max_t\left(\frac{1-y(t-1)^r}{y^at^{ra}-y^a(t-1)^{ra}}\right) \le \max_t\left(\frac{1-y(t-1)^r}{y^a (ra) t^{ra-1}}\right) \ . \] 
Thus, it suffices to prove that $\frac{v}{k} \geq \max_t f(t)$ where $f(t) := \left(\frac{1-y(t-1)^r}{y^a (ra) t^{ra-1}}\right)$. 
From the theorem statement we have $\frac{v}{k} \geq  f(t)$ holds for any $t \leq Z$; it remains to argue that the same is true when $t > Z$. 
Since we already know that $f(Z) \leq v/k$, it suffices to show that the function $f(\cdot)$ is decreasing over $[Z,\infty)$  i.e.,  $ f'(t) \leq 0$  for all $t \geq Z$.
%Since $y^at^{ra}-y^a(t-1)^{ra}= \int_{t-1}^t y^a (ra) x^{ar-1} \, \mathrm{d}x$ we have $y^a (ra) t^{ra-1}\le y^at^{ra}-y^a(t-1)^{ra}\le y^a (ra)(t-1)^{ra-1}$ and 
%\[\max_t\left(\frac{1-y(t-1)^r}{y^at^{ra}-y^a(t-1)^{ra}}\right) \geq \max_t\left(\frac{1-y(t-1)^r}{y^a (ra) (t-1)^{ra-1}}\right) \ . \] 
%Let $f(t) = \left(\frac{1-y(t-1)^r}{y^a (ra) (t-1)^{ra-1}}\right)$. 

We calculate the derivative $ f'(t)$ as follows
\begin{align*}
f'(t)&=-\frac{(t-1)^{r-1}t^{1-ra}y^{1-a}}{a}\\
&+\frac{(1-ra)t^{-ra}y^{-a}(1-(t-1)^{r}y)}{ra},
\end{align*}
so that $f'(t)\le 0$ if and only
\begin{align*}
\frac{(1-ra)y^{-a}(1-(t-1)^{r}y)}{rat^{ra}}\le\frac{y^{1-a}t(t-1)^{r-1}}{at^{ra}}\\
(1-ra)(1-(t-1)^ry)\le yt(t-1)^{r-1}r\\
(1-ra)\le y(t-1)^{r-1}((t-1)(1-ra)+tr).
\end{align*}
%\begin{align*}
%f'(t)&=-\frac{(t-1)^{r-1}(t-1)^{1-ra}y^{1-a}}{a}\\
%&+\frac{(1-ar)(t-1)^{-ra}y^{-a}(1-(t-1)^{r}y)}{ar},
%\end{align*}
%then $f'(t)\le 0$ if and only
%\begin{align*}
%\frac{(1-ar)y^{-a}(1-(t-1)^{r}y)}{ar(t-1)^{ra}}\le\frac{y^{1-a}(t-1)^r}{a(t-1)^{ra}}\\
%(1-ar)(1-(t-1)^ry)\le y(t-1)^{r}r\\
%(1-ar)\le y(t-1)^{r}(1-ar+r)
%\end{align*}
Since $(t-1)(1-ra)\le(t-1)(1-ra)+tr$, then the last expression certainly holds true if $(1-ra)\le y(t-1)^{r-1}(t-1)(1-ra)$ or equivalently, $\frac{1}{y}\le(t-1)^r$. 
%The last expression certainly holds true if $\frac{1}{y}\le(t-1)^r$. 
Since $Z:=\left\lceil 1+\left(\frac{1}{y}\right)^{1/r}\right\rceil$, it follows that $f'(t)\le 0$ for all $t \geq Z$.
\end{proof}

%\subsubsection{Discussion} There is strong evidence that the password distribution actually follows CDF-Zipf's law --- see section \ref{sec:ZipfLawYahoo} and~\cite{WangW16}. It is possible that the tail of the password distribution does not fit Zipf's law as it is not always possible to confidently estimate the true entropy of passwords that were observed with low frequency in a dataset. However, Zipf's law is consistent with empirical observations~\cite{WangW16}. We stress that even if CDF-Zipf's law does not fit the tail of the password distribution that $T(y,r)$ still characterizes adversary for as long as the distribution follows CDF-Zipf's law. In this case, whenever $v/k \geq T(y,r)$ we can claim that the adversary will crack {\em almost all} of the user's passwords perhaps excluding a few of the strongest passwords.

%% file: applications.tex
% !TEX root = sp-main.tex
\section{Analysis of Previous Password Breaches}
\label{sec:apps}
In this section, we apply our economic model to analyze the consequences of recent password breaches and the impact of defenses that could have been adopted. 
\subsection{Breaches}
We focus on the following breaches in our analysis:
%\vspace{-0.3cm}

\subsubsection{Yahoo!}Attackers stole password hashes for $500$ million Yahoo! users in 2014, though the breach was unknown to the general public until 2016~\cite{YahooBreach}. While Yahoo! used BCRYPT to hash passwords~\footnote{An earlier 2013 Yahoo! breach affected approximately $1$ billion Yahoo! users~\cite{YahooBillionBreach}. We focus on the 2014 breach because the breach occurred after Yahoo! upgraded their password hashing algorithm from MD5 to BCRYPT. We note that any negative findings about the 2014 breach will certainly extend to the earlier breach since a weaker hashing algorithm was involved. }, they have not publicly specified the number of hash iterations $\tau$ that they used. However, we do have empirical password frequency data from 70 million Yahoo! users which allowed us to derived CDF-Zipf parameters $y=0.0211$ and $r=0.2166$ for Yahoo! passwords. Thus, we can predict the $\%$ of cracked passwords for different values of $\tau$ that Yahoo! might have chosen.   

%\vspace{-0.2cm} 
\subsubsection{ Dropbox}  Attackers stole password hashes for $\approx 68.7$ million Dropbox users though the breach was unknown to the general public until 2016~\cite{DropboxBreach}. Dropbox used BCRYPT at level $8$ (i.e., $\tau = 2^8 = 256$ hash iterations) to hash passwords. We don't have empirical password data from Dropbox users from which we can derive Zipf's law parameters $y$ and $r$. However, we have Zipf's law parameters for many other datasets such as RockYou, Tianya, CSDN and Yahoo! allowing us to predict how many passwords a value $v$ adversary would crack if, say, Dropbox passwords and RockYou passwords have similar strength.  Arguably, Dropbox passwords could be quite valuable as they are often used to protect sensitive data.

\subsubsection{AshleyMadison}  Attackers stole nearly $40$ million AshleyMadison password hashes~\cite{AMWeak} in 2015 and released the stolen data publicly a month later. AshleyMadison primarily used  BCRYPT at level $12$ $(\tau=2^{12}=4,096$ hash iterations) to hash passwords~\cite{AMStrong}. However, CynoSure Prime noticed that some passwords were effectively protected with MD5 instead of BCRYPT due to an implementation error. CynoSure Prime managed to crack approximately $11$ million of these MD5 hashes in just $10$ days~\cite{AMWeak}, though it has been claimed that most of the passwords protected by BCRYPT are uncrackable~\cite{AMStrong}. Similar to Dropbox, we do not have Zipf's law parameters for AshleyMadison users. However, it is plausible to believe that these parameters are comparable to the parameters derived from other datasets such as Yahoo! or RockYou!

\subsubsection{LastPass}   LastPass was using PBKDF2-SHA256 with $\tau= 10^5$ rounds of iteration when they were breached in 2015. Similar to AshleyMadison and Dropbox breaches we don't currently have Zipf's law parameters for LastPass passwords though we can still predict how many passwords would be breached under the assumption that these passwords have similar strength to passwords in other datasets like RockYou or Yahoo! Arguably master passwords will be more valuable to an attacker than regular passwords as a master password will unlock multiple user accounts. On the other hand previous research~\cite{bonneau2012yahoo} has not found a clear correlation between password strength and account value. 

%While this might suggest that users would select stronger passwords for LastPass than in other datasets like RockYou or Yahoo! However, Bonneau~\cite{bonneau2012yahoo} found that password strength does not seem to be correlated with the importance of an account. 

\paragraph{Estimating $v$} As described in Section \ref{sec:prelim} the value $v$ represents the value per password when all passwords are released on the market. Thus, although the actual black market prices may vary with supply, the parameter $v$ is fixed. Our estimate of this value parameter will depend on the current black market price, and model parameter $a$ (diminishing returns). In Table \ref{tab:v_conversion} we show various estimates of $v$ obtained from multiple estimates of black market password prices. These estimates include measurements from Fossi \cite{passwordBlackMarket} and more recent estimates from \cite{stockley_2016}, which finds that Yahoo! passwords go for 0.70-1.20 USD on the black market. To obtain the estimates in Table \ref{tab:v_conversion}, we assume that the black market prices were observed when just 1\% of the passwords were on the market. This allows us to esimate the value $v$ if all passwords were to be released using equation \ref{eq:valueFormula}. We remark that the difference between the two estimates \cite{stockley_2016} and \cite{passwordBlackMarket} may be explained due to additional black market supply. We view $a=0.8$ as substantial diminishing returns e.g., the marginal revenue decreases by a factor of $1/3$ when the attacker compromises all accounts. An interesting direction for future work may be to estimate the parameter $a$ from a longitudinal study of black markets.
%It is worth noting that the 4 USD estimate from Fossi is larger than the estimates. It may be the case that since the large Yahoo! breach prices have fallen due to market forces, and we note that the $a=0.8$ value for Fossi's 4 USD estimate roughly matches our calculations.

\paragraph{Translating between $v$ and $v^{\$}$} Bonneau and Schechter~\cite{BonneauS14} observed that in 2013, Bitcoin miners were able to perform approximately $2^{75}$ SHA-256 hashes in exchange for bitcoin rewards worth about $\$257M$. Correspondingly, one can estimate the cost of evaluating a SHA-256 hash to be approximately $C_H=\$7\times10^{-15}$. Alternatively, the cost can be viewed as the economic opportunity cost of evaluating each hash function (for instance, renting a botnet or computing on a cloud platform.) Because Bitcoin mining is almost exclusively performed on application specific integrated circuits (ASICs) the above cost analysis implicitly assumes that the attacker is willing to fabricate an an ASIC to evaluate PBKDF2-SHA256 or BCYRPT.  We contend that this is a plausible scenario for a rational attacker, since fabrication costs would amortize over the number of user accounts being attacked (e.g., $500+$ million). Furthermore, we note that an attacker who is not willing to pay to fabricate an ASIC could obtain similar performance gains using a field programmable gate array (FPGA).

\subsection{Results}

In section~\ref{sec:ZipfLawAnalysis} we showed that, if passwords follow CDF-Zipf's law with parameters $y$ and $r$, and $v/k \geq T(y,r,a)$ then a rational adversary will crack $100\%$ of user passwords. Figure~\ref{fig:cdf2} plots $v = k\times T(y,r,0.8)$ for various thresholds from Table \ref{tab:MinThreshold} including Yahoo! and RockYou. Thus, for a point $(v,\tau)$ lying on the blue line, a value $v$ rational adversary will crack $100\%$ of Yahoo! passwords when he can compute the hash function at cost $k=\tau$. Note that $\tau=k$ for hash functions like BCRYPT and PBKDF2  --- the ones used by Yahoo!, Dropbox, AshleyMadison and LastPass. For reference, Figure~\ref{fig:cdf2} includes the actual values of $\tau$ selected by AshleyMadison, Dropbox and LastPass as well as the value $\tau=10^7$. Bonneau and Schechter estimated that SHA256 can be evaluated $10^7$ times in $1$ second on a modern CPU~\cite{BS14}. Thus, $10^7$ upper bounds the value of $\tau$ that one could select without delaying authentication for more than $1$ second when using PBKDF2-SHA256.

\begin{center}
\begin{table}
\begin{tabular}{|c|c|c|c|c|}
\hline
$R(t_{1\%})$ (USD) & a = 0.8 & a = 0.9 & a = 1.0 \\ \hline
0.70	& 0.28	&	0.44	&	0.70	\\ \hline
1.20	& 0.48	&	0.76	&	1.20	\\ \hline
4.00	& 1.59	& 	2.52	&	4.00	\\ \hline
30.00	& 11.94	& 	18.93	&	30.00	\\ \hline

\end{tabular}
\caption{$v$ conversion chart}
\label{tab:v_conversion}
\end{table}
\end{center}
The plots predict that, unless we set $\tau \gg 10^7$, the adversary will crack $100\%$ of passwords in almost every instance. In particular, the levels of key-stretching performed by Dropbox, AshleyMadison and even Lastpass are all well below the thresholds necessary to protect Yahoo!, RockYou or CSDN passwords.

Figure~\ref{fig:cdf3} is similar to Figure~\ref{fig:cdf2} except that we rescale to $y$ axis to show $v^{\$}$, given monetary estimations of computation cost and password values, so that we can focus on the number of hash iterations necessary to simply avoid all passwords being cracked.

\begin{figure*}[!h]
\centering
\subfigure[$v/k = T(y,r,0.8)$ for RockYou, CSDN and Yahoo!]{
\centering
\begin{tikzpicture}[scale=0.63] 
\begin{axis}[title style={align=center},
   xlabel={$\log_2(\tau)$},
   ylabel={$v(\times10^{16})$},
	 xmin={7},
	 xmax={24},
	 ymin={0},
	 ymax={0.05},
   ylabel shift = -3pt,
   grid=major,
	 cycle list = {{red, mark=circle},  {orange, mark=none}, {purple, mark=none}},
   legend style = {font=\small, at={(.05,.95)}, anchor=north west},
	 legend entries = {RockYou, CSDN, Yahoo!, $v^\$=\$0.28$ (estimate),$v^\$=\$0.48$ (estimate), Dropbox $\tau$, AshleyMadison $\tau$, NIST $\tau$ (min), LastPass $\tau$, $\tau=10^7$ (1sec)}
  ]

\addplot[color=black, mark=r, domain=7:24]{2.03629*10^7*(2^x)/10^16};
%CSDN
\addplot[color=red, mark=-, domain=7:24]{9.2439*10^7*(2^x)/10^16};
%enter yahoo T(y,r) value here:
\addplot[color=blue, mark=o, domain=7:24]{2.68677*10^7*(2^x)/10^16};
\addplot[style=dotted, domain=7:24]{0.057143*0.28/4};
\addplot[style=dashed, domain=7:24]{0.057143*0.48/4.0};
\addplot[mark=x] coordinates{(8, 0)  (8,0.006) (8,0.1/4) (8,0.15/4) (8, 0.2/4)};
\addplot[mark=triangle] coordinates{(12, 0)(12,0.006)  (12,0.1/4) (12,0.15/4) (12, 0.2/4)};
\addplot[color=red,mark=star] coordinates {(13.2877,0) (13.2877,0.006) (13.2877,0.1/4) (13.2877,0.15/4) (13.2877,0.2/4)};
\addplot[mark=square] coordinates{(16.66667, 0) (16.66667,0.05/4) (16.66667,0.006) (16.66667,0.15/4) (16.66667, 0.2/4)};
\addplot[color=red,mark=diamond] coordinates {(23.2635,0) (23.2635,0.006) (23.2635,0.1/4) (23.2635,0.15/4) (23.2635,0.2/4)};
\end{axis}
\end{tikzpicture}
%\caption{$v/k = T(y,r)$ for RockYou and Yahoo! plus $\tau=k$ for Dropbox, AshleyMadison and LastPass.}
\label{fig:cdf2}}
\subfigure[$v^{\$}$ vs. $\tau$ for $v = k \times T(y,r,0.8)$. ]{

\begin{tikzpicture}[scale=0.63] 
\begin{axis}[title style={align=center},
  xlabel={$\log_2(\tau)$},
  ylabel={$v^\$$},
	xmin={7},
	xmax={24},
	ymin={0},
	ymax={0.5},
  ylabel shift = -3pt,
  grid=major,
  legend style = {font=\small, at={(.05,.95)}, anchor=north west},
	legend entries = {RockYou, CSDN, Yahoo!, $v^{\$}=\$0.28$ (estimate), $v^{\$}=\$0.48$ (estimate),$\tau=10^7$ (1sec), Dropbox $\tau$, AshleyMadison $\tau$, NIST $\tau$ (min), LastPass $\tau$}
  ] 
\addplot[color=black, domain=7:24]{(2.03629*10^7)*(7*10^-15)*(2^x)};
%CSDN
\addplot[color=red, mark=-,domain=7:24]{(9.2439*10^7)*(7*10^-15)*(2^x)};
%Enter yahoo T(y,r) value here:
\addplot[color=blue, mark=o, domain=7:24]{(2.68677*10^7)*(7*10^-15)*(2^x)};
\addplot[color=black, style=dotted, domain=7:24]{0.28};
\addplot[color=black, style=dashed, domain=7:24]{0.48};
\addplot[color=red,mark=diamond] coordinates {(23.2635,0) (23.2635,0.05) (23.2635,0.5/2) (23.2635,0.5)};
\addplot[color=red,mark=star] coordinates {(13.2877,0) (13.2877,0.05) (13.2877,0.5/2) (13.2877,0.5)};
\addplot[mark=x] coordinates{ (8, 0)  (8,0.05) (8,0.5/2) (8, 0.5)};
\addplot[mark=triangle] coordinates{ (12, 0)(12,0.05)  (12,0.5/2) (12, 0.5)};
\addplot[mark=star] coordinates {(13.2877,0) (13.2877,0.05) (13.2877,0.5/2) (13.2877,0.5)};
\addplot[mark=square] coordinates{ (16.66667, 0) (16.66667,0.05) (16.66667,0.5/2)  (16.66667, 0.5)};
%rockyou coordinates
%\addplot[mark=*] coordinates {(25.00545,4)};
%\addplot[mark=*] coordinates {(27.91234,30)};
%tianya coordinates
%\addplot[mark=*] coordinates {(24.57985,4)};
%\addplot[mark=*] coordinates {(27.48674,30)};
%yahoo coordinates
%solve((2.25435*10^7)*(7*10^-15)*(2^x)=30)
%\addplot[mark=*] coordinates {(24.5954,4)};
%\addplot[mark=*] coordinates {(27.5022,30)};
\end{axis} 
\end{tikzpicture}
%\caption{$v^{\$}$ vs. $\tau$ for $v = k \times T(y,r)$. }
\label{fig:cdf3}}
\subfigure[$v^{\$}$ versus $T(y,r,0.8)$ when $v = k\times T(y,r,0.8)$, at fixed values of $k$]{
\begin{tikzpicture}[scale=0.63] 
\begin{axis}[title style={align=center},
   xlabel={$\log_{10}($T(y,r,0.8)$)$},
   ylabel={$v^{\$}$},
	 xmin={5},
	 xmax={12},
	 ymin={0},
	 ymax={1},
	 %extra x ticks = {7.04034, 6.741211, 5.728109, 7.35503},
	 %extra x tick labels = {},
	 %extra x tick style = {major grid style=red, tick align=outside, tick style=red},
   ylabel shift = -3pt,
   grid=major,
	 cycle list = {{red, mark=none},  {orange, mark=none}, {yellow, mark=none}, {purple, mark=none}},
	 legend style = {font=\small, at={(.05,.95)}, anchor=north west},
	 legend entries = {AshleyMadison ($k=2^{12}$), Dropbox ($k=2^8$), LastPass $(k=10^5)$, NIST MIN $(k=10^4)$, $v^{\$}=\$0.28$, $v^{\$}=\$0.48$, RockYou, CSDN, Yahoo!}
  ] 
\addplot[color=green,mark=o, domain=5:12]{(10^x)*(7*10^-15)*(2^12)};
\addplot[color=blue, mark=diamond, domain=5:12]{(10^x)*(7*10^-15)*(2^8)};
\addplot[color=black, domain=5:11]{(10^x)*(7*10^-15)*(5000+10^5)};
\addplot[color=red, mark=star, domain=5:12]{(10^x)*(7*10^-15)*(10^4)};
\addplot[color=black, style=dotted, domain=5:12]{0.28};
\addplot[color=black, style=dashed, domain=5:12]{0.48};

%rockyou
\addplot[mark=x] coordinates{ (7.308840, 0) (7.308840, 0.1) (7.308840, 1)};
%CSDN
\addplot[mark=triangle, color=red] coordinates{ (7.965855, 0) (7.965855, 0.1)  (7.965855, 1)};
%insert log(T(y,r)) base 10 for Yahoo here
\addplot[mark=square, color=blue] coordinates{ (7.42923, 0) (7.42923, 0.1) (7.42923, 1)};

\end{axis} 
\end{tikzpicture}
%\caption{$v^{\$}$ versus $T(y,r)$ when $v = k\times T(y,r)$, at fixed values of $k$ selected by AshleyMadison, Dropbox and LastPass. Vertical lines show thresholds $T(y,r)$ for  RockYou and Yahoo! }
\vspace{-0.2cm} 
\label{fig:cdf4}}
\caption{($a=0.8$)}
\label{mainFigure}
\end{figure*}

While we do not have CDF-Zipf parameters for other breaches such as AshleyMadison, Dropbox, or LastPass, we do have the value $\tau=k$ for each of these breaches. Figure~\ref{fig:cdf4} plots $v = k \times T(y,r,0.8)$ only this time we hold $k$ constant and allow $T(y,r,0.8)$ to vary. For example, in the black line we fix $k=\tau=10^5$ since LastPass used PBKDF2-SHA256 with $\tau=10^5$ hash iterations and allow $T(y,r,0.8)$ to vary. The vertical lines represent the thresholds $T(y,r,0.8)$ we derive from CDF-Zipf's law fits for RockYou, Tianya and Yahoo!    Table \ref{tab:MinThreshold} shows the value of $T(y,r,0.8)$ obtained from $10$ different password datasets. Observe that in all of cases we had $T(y,r,0.8) \leq 7.64 \times 10^7$. As in Figure \ref{fig:cdf3} the $y$-axis in Figure~\ref{fig:cdf4}  is scaled to show the value $v^\$$ in USD (estimated). Thus, if Dropbox (resp. AshleyMadison/LastPass) passwords have comparable strength to Yahoo! passwords (resp. Tianya, RockYou) then a rational adversary would crack $100\%$ of these passwords. Indeed, Figure~\ref{fig:cdf4} shows that unless the thresholds $T(y,r,a)$ for Dropbox/LastPass/AshleyMadison are significantly larger than the previously observed thresholds, a rational adversary would be compelled to crack all passwords, given the range of password values. For example, even if the threshold $T(y,r,a)$ for Dropbox exceeds the threshold for Yahoo! by four orders of magnitude then the adversary will still crack $100\%$ of these passwords.

\subsection{Discussion} \label{subsec:DiscussionTail}
Figures \ref{fig:cdf2}, \ref{fig:cdf3} and \ref{fig:cdf4} paint a grim picture. PBKDF2 and BCRYPT most likely provide dramatically insufficient protection for most AshleyMadison, Dropbox, Yahoo! and LastPass users --- even if we used the lowest estimation of the value parameter $v$ from Table \ref{tab:v_conversion} ($v^{\$} = 0.28$ USD) and we assume that the attacker faces substantial diminishing returns $(a=0.8)$ for additional cracked passwords. Furthermore, it would not have been possible to provide sufficient protection for users using PBKDF2 or BCRYPT without introducing intolerable authentication delays ($\geq1$ second).

Our analysis assumes that the password distribution truly follows CDF-Zipf's law. While previous research~(e.g., \cite{WangZipfLaw14,WangW16} and our own results in Section \ref{sec:ZipfLawYahoo}) strongly supports the hypothesis that {\em most} of the password distribution follows Zipf's law, it is not possible to definitively state that the tail of the password distribution does not follow Zipf's law since each of the passwords in the tail were (by definition) observed with low frequency. We stress that even if CDF-Zipf's law does not fit the tail of the password distribution that $T(y,r,a)$ still characterizes adversary behavior. For example, suppose that the $(100-x)\%$ of passwords follow a  Zipf's law distribution with parameters $y,r$ while $x\%$ of passwords in the tail of the password distribution do not.  In this case, whenever $v/k \geq T(y,r,a)$ we a rational adversary will crack {\em at least} $(100-x)\%$ of the user's passwords which follow Zipf's Law. 

 %However, there is strong evidence that the password distribution actually follows CDF-Zipf's law with the possible exception of the tail --- see section \ref{sec:ZipfLawYahoo}.  it is not possible to definitively state that the tail of the password distribution does not follow Zipf's law since each of the passwords in the tail were (by definition) observed with low frequency. 

%\subsubsection{Discussion} There is strong evidence that the password distribution actually follows CDF-Zipf's law --- see section \ref{sec:ZipfLawYahoo} and~\cite{WangW16}. It is possible that the tail of the password distribution does not fit Zipf's law as it is not always possible to confidently estimate the true entropy of passwords that were observed with low frequency in a dataset. However, Zipf's law is consistent with empirical observations~\cite{WangW16}.  In this case, whenever $v/k \geq T(y,r)$ we can claim that the adversary will crack {\em almost all} of the user's passwords perhaps excluding a few of the strongest passwords.

\subsection{Memory Hard Functions}
Memory hard functions potentially provide a way of increasing computation cost without drastically increasing computation time. As the name suggests memory hard functions require a large amount of memory to evaluate. Thus, the cost of purchasing/renting hardware for password cracking, approximated by a functions Area x Time (AT) complexity, can be substantial for an attacker. Specifically, AT complexity of SCRYPT~\cite{Per09}, scales quadratically with the number of time steps~\cite{ACKKPT16}. Thus, as discussed in Section~\ref{sec:prelim}, we estimate $k^{\$}=\tau C_H+\tau^2 C_{mem}$, where  $C_H \approx \$7\times10^{-15}$~\cite{BonneauS14} and $C_{mem}\approx\frac{C_H}{3000}$ as in \cite{BK15,AB16}. 

In the last section we assumed that the attacker faced aggressive diminishing marginal returns for additional cracked passwords and we used the lowest possible estimations of adversary value finding that an attacker still cracks $100\%$ of passwords from a Zipf's law distribution. By contrast, in this section we operate under the conservative assumptions that the attacker does not face diminishing returns and we use the larger estimations of adversary value in our analysis. Nevertheless, we find that the use of MHFs can substantially reduce the $\%$ of cracked passwords.   

Figure~\ref{fig:cdf:mhf} plots $v^\$$ (estimate) versus the minimum value of $\tau$ necessary to prevent a rational attacker from cracking $100\%$ of passwords. For example, the blue line predicts that if Yahoo! had adopted memory hard functions with only $\tau=2^{20}$ iterations ($0.1$ seconds) then a value $\$30$ adversary will not crack all passwords selected from a CDF-Zipf's law distribution with the parameters $y=0.0211$ and $r=0.2166$, the parameters for our CDF-Zipf's fit for Yahoo! passwords. By contrast, Yahoo! would need to set $\tau=2^{26}$ ($\approx 7$ seconds) when using a function like PBKDF2 or BCRYPT just to ensure that the adversary does not crack $100\%$ of passwords when $a=1$.
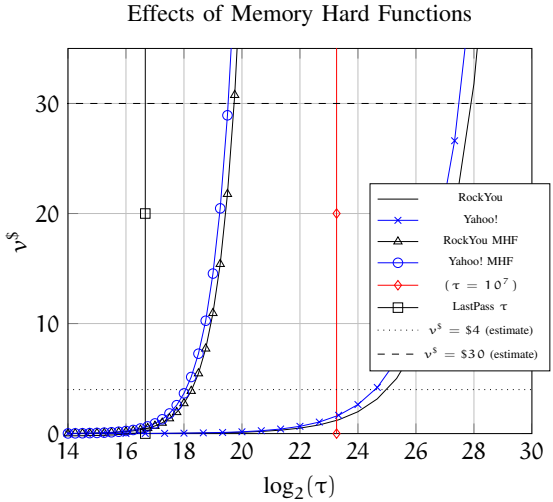
\begin{figure}[htb]
\centering
\begin{tikzpicture}[scale=0.9] 
\begin{axis}[title style={align=center},
	 title={Effects of Memory Hard Functions},
   xlabel={$\log_2(\tau)$},
   ylabel={$v^{\$}$},
	 xmin={14},
	 xmax={30},
	 ymin={0},
	 ymax={35},
   ylabel shift = -3pt,
   grid=major,
   legend style = {font=\tiny, at={(.65,.65)}, anchor=north west},
	 legend entries = {RockYou, Yahoo!, RockYou MHF, Yahoo! MHF, $(\tau=10^7)$,  LastPass $\tau$, $v^{\$}=\$4$ (estimate), $v^{\$}=\$30$ (estimate)}	
  ] 

\addplot[color=black, domain=14:30]{(1.69657*10^7)*7*(10^-15)*(2^x)};
%tianya
%\addplot[color=red, domain=14:30]{(2.27871*10^7)*7*(10^-15)*(2^x)};
%enter yahoo T(y,r) below
\addplot[color=blue, mark=x, domain=14:30]{(2.25435*10^7)*7*(10^-15)*(2^x)};
\addplot[color=black, mark=triangle, domain=14:20]{(1.69657*10^7)*((7*10^-15)*(2^x)+7/3*(10^-18)*(2^x)^2)};
%tianya
%\addplot[color=blue, domain=14:20]{(2.27871*10^7)*((7*10^-15)*(2^x)+7/3*(10^-18)*(2^x)^2)};
%enter yahoo T(y,r) below
\addplot[color=blue, mark=o, domain=14:20]{(2.25435*10^7)*((7*10^-15)*(2^x)+7/3*(10^-18)*(2^x)^2)};
\addplot[color=red,mark=diamond] coordinates {(23.2635,0) (23.2635,20) (23.2635,40) (23.2635,50)};

\addplot[mark=square] coordinates{ (16.66667, 0) (16.66667,20) (16.66667,40)  (16.66667, 50)};
\addplot[color=black, style=dotted, domain=0:30]{4};
\addplot[color=black, style=dashed, domain=0:30]{30};

\end{axis} 
\end{tikzpicture}
\caption{Memory Hard Functions: $v^\$$ vs $\tau$ when $v = k\times T(y,r,1)$ using thresholds $T(y,r,1)$ for RockYou and Yahoo! $k=\tau C_H+\tau^2 C_{mem}$ for MHFs and $k=C_H \times \tau$ otherwise.  }
\label{fig:cdf:mhf}
\end{figure}

Figure \ref{fig:cdf:mhf} predicts that MHFs prevent a rational adversary from cracking {\em all} passwords from a Zipf's law distribution. Of course, if the adversary still cracks $99.9\%$ of passwords then this result would not be particularly exciting.  Figure \ref{fig:cdf:mhf2} plots $\%$ cracked passwords vs. $\tau$ against a value $v^\$=\$4$ adversary. These plots provide an optimistic outlook for MHFs. For example, the plots predict that we can significantly reduce the $\%$ of cracked passwords (easily below $20\%$) with out introducing unacceptably long authentication delays when passwords follow a Zipf's law distribution. By contrast, the plots predict that we would need to set $\tau \approx 2^{32}$ ($400+$ seconds) to achieve the same result using PBKDF2 or BCRYPT when $a=1$.

\begin{figure}[htb]
\centering
\begin{tikzpicture}[scale=0.9] 
\begin{axis}[title style={align=center},
	 title={Effects of Memory Hard Functions},
   xlabel={$\log_2(\tau)$},
   ylabel={$\%$ cracked},
	 xmin={18},
	 xmax={37},
	 ymin={0},
	 ymax={100},
   ylabel shift = -3pt,
   grid=major,
	 cycle list = {{black, mark=triangle}, {blue, mark=o},  {black, mark=none}, {blue, mark=x}},
   legend style = {font=\tiny, at={(.99,.95)}, anchor=north east},
	 legend entries = {RockYou MHF, Yahoo! MHF, RockYou, Yahoo!, $\max \#$ SHA256 hashes in 1sec $(10^7)$}
  ] 
%solve((1-0.037433*(10^1-1)^0.187227)/(0.037433*0.187227*(10^1)^(0.187227-1)) = 4 / ((y*7*10^-15)+(2.33*10^-18)*y^2), y)
%log(4.42*10^7)/log(2)
%0.037433*x^0.187227 at x=10^1
%solve((1-0.0211*(10^1-1)^0.2166)/(0.0211*0.2166*(10^1)^(0.2166-1)) = 4 / ((y*7*10^-15)+(2.33*10^-18)*y^2), y)
\addplot coordinates{(25.4, 5.76078) (24.07, 8.86561) (22.76, 13.6438) (21.477, 20.9973) (20.23, 32.3139) (19.101, 49.7298) (18.297, 76.5321)};
\addplot coordinates{(25.1241, 5.76078) (23.8403, 8.86561) (22.56785, 13.6438) (21.3164, 20.9973) (20.105, 32.3139) (18.9827, 49.7298) (18.135, 76.5321)};
%solve((1-0.037433*(10^7-1)^0.187227)/(0.037433*0.187227*(10^7)^(0.187227-1)) = 4 / ((y)*7*10^-15), y)
\addplot coordinates{(36.5977, 8.86561) (33.985, 13.6438) (31.403, 20.9973) (28.92875, 32.3139) (26.65676, 49.7298) (25.0567, 76.5321) (25, 100)};
%solve((1-0.0211*(10^7-1)^0.2166)/(0.0211*0.2166*(10^7)^(0.2166-1)) = 4 / ((y)*7*10^-15), y)
\addplot coordinates{(36.12807, 8.86561) (33.5836, 13.6438) (31.08165, 20.9973) (28.6616, 32.3139) (26.42107, 49.7298) (24.73308, 76.5321)};
\addplot[color=red,mark=diamond] coordinates {(23.2635,0) (23.2635,33) (23.2635,67) (23.2635,100)};
\end{axis}
\end{tikzpicture}
\caption{Memory Hard Functions: $\%$ cracked by value $v=\$4$ adversary against MHF with running time parameter $\tau$. }
\label{fig:cdf:mhf2}
\end{figure}
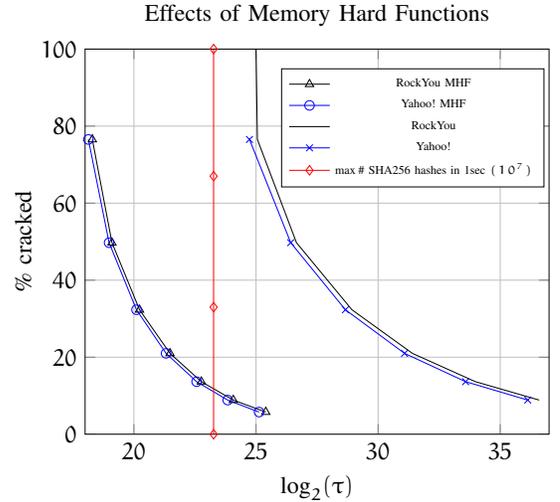

%% file: modelIndependent.tex
% !TEX root =sp-main.tex
\section{Model Independent Analysis} \label{sec:ModelIndependent}
In this section we derive model-independent upper and lower bounds on the $\%$ of users whose passwords would be cracked by a rational adversary. The advantage of a model independent analysis is that the bounds we derive apply even if we do not make any assumptions about the shape of the password distribution. As we observed previously it is not possible to definitively claim that the tail of the password distribution follows Zipf's law --- even if the tail of the distribution is not known to be {\em inconsistent} with Zipf's law~\cite{WangZipfLaw14,WangW16}. The disadvantage of a model independent analysis is that the bounds we are able to derive may not always be tight as the bounds we may be able to derive using specific modeling assumptions e.g., Zipf's law. In this section we assume for the sake of simplicity that $a=1$ i.e., the marginal value of each additional cracked password remains constant.  

Suppose that we are given $N$ independent samples $pwd^1$,$\ldots$, $pwd^N$ $\leftarrow \mathcal{X}$ from an (unknown) distribution $\mathcal{X}$. As before, we will let $f_i$ denote the number of users who chose password $pwd_i$ in a dataset and without loss of generality assume that these frequencies are sorted so that $f_i \geq f_{i+1}$. We can use $f_i$ to obtain an estimate $\hat{p}_i = \frac{f_i}{N}$ for $p_i$, the true probability that that a random user selects the password $pwd_i$. While we do have $\hat{p}_i \geq \hat{p}_{i+1}$ we stress that we may no longer assume that $p_i \geq p_{i+1}$ since our empirical value $\hat{p}_i$ (resp. $\hat{p}_{i+1}$) may over/under estimate the true probability $p_i$. 

\subsection{Lower Bound}
Theorem \ref{thm:ModelIndependentLB1} lower bounds the number of passwords that will be cracked by a rational adversary in expectation. The expectation is taken over $N$ passwords sampled from $\mathcal{X}^N$.

\newcommand{\thmLowerBound}{If $\frac{V}{k} \geq NL$ and $a=1$ then a rational adversary will crack at least 
\[ \sum_{i: f_i\geq j} f_i - \frac{N}{(j-1)!L^{j-1}} \  \]
user passwords, in expectation.}
\begin{theorem} \label{thm:ModelIndependentLB1}
\thmLowerBound
\end{theorem}

The proof of Theorem \ref{thm:ModelIndependentLB1} is in appendix \ref{apdx:MissingProofs}. The proof begins with the observation that a password $pwd^t=pwd_i$ will certainly be cracked by a value $V$ adversary if $p_i \geq \frac{1}{NL}$. We then introduce the notion of a {\em $(j,L)$-bad overestimate}. In particular, a $(j,L)$-bad overestimate for $pwd^t$  occurs when $p_i < \frac{1}{NL}$ but $f_i \geq j$. If we have $f_i \geq j$ then either $p_i \geq \frac{1}{NL}$ and the password will be cracked, or we have a $(j,L)$-bad overestimate for the password $pwd^t$. We can then show that $\frac{N}{(j-1)!L^{j-1}}$ upper bounds the expected number of passwords $pwd^t$ with $(j,L)$-bad overestimates.

\subsection{Upper Bound}
In contrast to Theorem \ref{thm:ModelIndependentLB1} , Theorem \ref{thm:ModelIndependentUB1} upper bounds the $\%$ of passwords that we expect an attacker to compromise. 
\newcommand{\themModelIndepUpperBound}{If $\frac{V}{k} \le NL\left(1-\frac{1+\eps}{N}\sum_{i=1}^t f_i\right)$ for fixed $\eps$ and $t$, then except with probability $\exp\left(-\frac{\eps^2N\sum_{i=1}^t p_i}{2(1+\eps)^2}\right)$, a rational adversary will crack at most $\sum_{i: f_i > j} f_i  + \mu(N,L,j)$ user passwords where $\mu(N,L,j) = $
\[ \displaystyle \sum_{i: 0 < f_i \leq j} f_i \sum_{\ell=0}^{j-1}\binom{N-1}{\ell} \left( \frac{1}{NL}\right)^\ell\left(\frac{NL-1}{NL}\right)^{N-\ell-1}  \ .\]
}
\begin{theorem} \label{thm:ModelIndependentUB1}
\themModelIndepUpperBound
\end{theorem}

The proof of Theorem \ref{thm:ModelIndependentUB1} is in appendix \ref{apdx:MissingProofs}. Briefly, we apply Chernoff bounds to show that, if $\frac{V}{k} \leq NL\left(1-\frac{1+\eps}{N}\sum_{i=1}^t f_i\right)$, then with high probability the number of user passwords in our dataset that a rational adversary cracks is at most 
\[ \sum_{i: f_i\geq j} f_i + \sum_{i} f_i \times C_i \  . \]
Here, $C_i$ denotes the event that we have a $(j,L)$-bad underestimate for the password $pwd_i$. We then separately upper bound the sum $\sum_{i} f_i \times C_i$ to obtain the bound in Theorem \ref{thm:ModelIndependentUB1}.

\subsection{Applications}
Theorems \ref{thm:ModelIndependentLB1} and \ref{thm:ModelIndependentUB1}  allows us to derive different upper and lower bounds by plugging in different values of $j$ and $L$. For example, by increasing $j$ we decrease the term $\frac{N}{(j-1)!L^{j-1}}$ in Theorem \ref{thm:ModelIndependentLB1}, but we also decrease the sum $\sum_{i:f_i \geq j} f_i$. Increasing (resp. decreasing) $L$ is equivalent to assuming the adversary has a higher (resp. lower) value for cracked passwords, which intuitively allows us to establish higher lower bounds (resp. smaller upper bounds) on the percentage of passwords cracked. 

\subsubsection{Lower Bounds} Applying Theorem \ref{thm:ModelIndependentLB1} we  can derive specific lower bounds for each of the datasets studied by \cite{WangW16} as well as for the Yahoo! frequency corpus. For most datasets we obtain our lower bound by setting $j=2$ and $L=10$. For the Yahoo! and RockYou datasets we obtained better lower bounds by setting $j=3$ and $L=10$. The result appears below:

%5 212345
%4 355243
%3 713732      $ 0.7 million passwords that were each chosen by exactly 3 users (2.1e6 users)
%2 2261395    $ 3.3 million passwords that were each chosen by exactly 2 users (6.6e6 users)
%1 29452171  % 29.4 million unique passwords
%(1 - 14326970 / 32581870 - 1/10)
\[
\resizebox{8cm}{!}
{
\begin{tabular}{|c|cc|c|c|}\hline
Dataset & Unique PWs & Total PWs & $\frac{V}{k}$ & $\%$ cracked \\\hline
RockYou & 14,326,970 & 32,581,870 & $3.2582\times10^8$ & 46.03 \\\hline
000webhost & 10,583,709 & 15,251,073 & $1.5251\times10^8$ & 20.60 \\\hline
Battlefield & 417,453 & 542,386 & $5.4239\times10^6$ & 13.03 \\\hline
Tianya & 12,898,437 & 30,901,241 & $3.0901\times10^8$ & 48.26 \\\hline
Dodonew & 10,135,260 & 16,258,891 & $1.6259\times10^8$ & 27.66 \\\hline
CSDN & 4,037,605 & 6,428,277 & $6.4283\times10^7$ & 27.19 \\\hline
Mail.ru & 2,954,907 & 4,932,688 & $4.9327\times10^7$ & 30.01 \\\hline
Gmail & 3,132,028 & 4,929,090 & $4.9291\times10^7$ & 26.46 \\\hline
Flirtlife.de & 115,589 & 343,064 & $3.3406\times10^6$ & 56.04 \\\hline
\textcolor{blue}{Yahoo!} & \textcolor{blue}{$2.94\times10^7$} & \textcolor{blue}{$7\times10^7$} & \textcolor{blue}{$7\times10^8$} & \textcolor{blue}{51} \\\hline
\end{tabular}

}
\]

{\bf \noindent Remark: } When $j=1$ we have $\sum_{i: f_i\geq j} f_i - \frac{N}{(j-1)!L^{j-1}} = N - N = 0$ meaning that Theorem \ref{thm:ModelIndependentLB1} provides no lower bound on the $\%$ of cracked passwords. 
At first glance this may appear to be a shortcoming of the theorem. 
However, we observe that it is impossible to obtain better lower bounds without making assumptions about the password distribution. 
In particular, let $\mathcal{X}_1$ (resp. $\mathcal{X}_2$) be the uniform distribution over a set of $2^{3n}$ (resp. $2^{8n}$) passwords. 
Observe that $\mathcal{X}_1$ and $\mathcal{X}_2$ can induce dramatically different rational attacker behavior (e.g., if the value of a password is $2^{3n}k$, the adversary will crack $100\%$ of passwords if the true password distribution is $\mathcal{X}_1$ and $0\%$ of passwords if the true distribution is $\mathcal{X}_2$). 
However, if we draw $N=2^n$ samples from $\mathcal{X}_1$ and $\mathcal{X}_2$, then the frequency lists for the two password distributions will be indistinguishable ($f_1=f_2=\ldots=f_N=1$) by birthday bounds ($N\ll 2^{1.5n}$). 
%However, a rational adversary with a value $v=2^{3n} k$ will always crack $100\%$ of passwords from distribution $\mathcal{X}_1$ (since $MR(i) \geq vp_i \geq k$ for each $i$) and $0\%$ of passwords from distribution $\mathcal{X}_2$ (since $p_iv = 2^{-8n}2^{3n}k \ll k$). 

\subsubsection{Upper Bounds} Similarly, we may use Theorem \ref{thm:ModelIndependentUB1} to derive model independent upper bounds on the percentage of Yahoo! passwords cracked by a rational adversary as shown in Figure~\ref{fig:mia:cracked}. As Figure~\ref{fig:mia:cracked} shows we could potentially use memory hard functions to reduce the $\%$ of cracked passwords to $\approx 20\%$ without increasing authentication time past $1$ second. This is particularly, impressive when one considers that an attacker only needs a single guess to achieve success rate $1\%$! 
%\jnote{When $n=2^{20}$ (0.1 s) we have $k^{\$} = 2.57\times 10^{-6}$. Setting $V=\$30$ we get $V/k \approx 1.17 \times 10^{7}$. Increasing $n=2^{21}$ we have $k^{\$} = 2^2\times 2.57\times 10^{-6}$ and $V/k=2.92\times 10^6$. At $n=2^{22}$ we get $k^{\$} = 2^4\times 2.57\times 10^{-6}$ and  $7.31\times 10^{5}$. Finally, when $n=2^{23}$ (0.8 s) we have $k^{\$} = 2^6\times 2.57\times 10^{-6}$ and $V/k=1.83\times 10^{-5}$. When $V = \$4$ and $n=2^{23}$ we get $V/k \approx 2.44 \times 10^{4}$.}
\ignore{\[
\begin{tabular}{|c|c|}\hline
$V/k$ & $\%$ cracked \\\hline
$10^8$ & $100$ \\\hline
$5\times10^7$ & $99$ \\\hline
$10^7$ & $61.38$ \\\hline
$5\times10^6$ & $56.53$ \\\hline
$10^6$ & $52.42$ \\\hline
$5\times10^5$ & $42.64$ \\\hline
$10^5$ & $37.46$ \\\hline
$5\times10^4$ & $26.30$ \\\hline
$10^4$ & $22.24$ \\\hline
\end{tabular}
\]}

\begin{table}
\begin{tabular}{|c|c|c|c|c|c|c|c|c|c|}\hline
$V/k$             & $10^8$    & $5\times10^7$ & $10^7$   & $5\times10^6$ & $10^6$ \\\hline 
$\%$ cracked  & $100$      & $99$               & $61.38$  & $56.53$           & $52.42$  \\\hline\hline
$V/k$ &   & $5\times10^5$  & $10^5$  & $5\times10^4$ & $10^4$  \\\hline
$\%$ cracked & & $42.64$           & $37.46$  & $26.30$           & $22.24$  \\\hline
\end{tabular}
\caption{Model Independent Upper Bound $\%$ cracked}
\end{table}

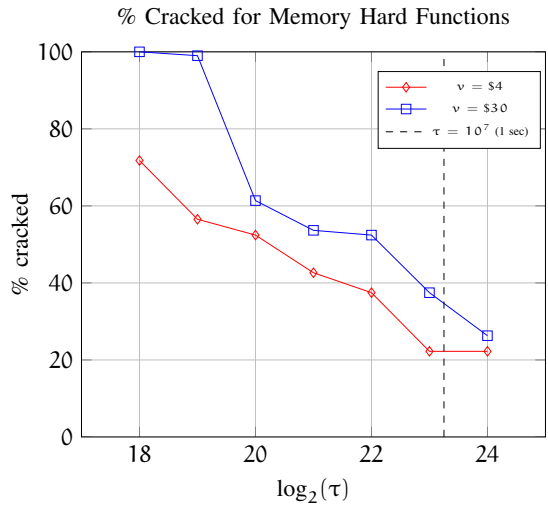
\begin{figure}[htb]
\centering
\begin{tikzpicture}[scale=0.9] 
\begin{axis}[title style={align=center},
	 title={$\%$ Cracked for Memory Hard Functions},
   xlabel={$\log_2(\tau)$},
   ylabel={$\%$ cracked},
	 xmin={17},
	 xmax={25},
	 ymin={0},
	 ymax={100},
   ylabel shift = -3pt,
   grid=major,
   legend style = {font=\tiny, at={(.99,.95)}, anchor=north east},
	 legend entries = {$v=\$4$, $v=\$30$, $\tau=10^7$ (1 sec)}
  ] 
%30/((2^24)*7*10^-15+(2^24)^2*10^-15*7/3000)
%run program
%[0.2223594, 0.2223594, 0.3746360285714286, 0.4263993, 0.5241556013459862, 0.5652648555459768, 0.7177611031124516]
%[1, 0.9900191, 0.6137915769743909, 0.5365185845088717, 0.5241556013459862, 0.3746360285714286, 0.2630084]
\addplot[color=red,mark=diamond] coordinates {(24,22.236) (23,22.236) (22,37.464) (21,42.64) (20,52.42) (19,56.53) (18,71.78)};
\addplot[color=blue,mark=square] coordinates {(24,26.3) (23,37.464) (22,52.42) (21,53.65) (20,61.38) (19,99) (18,100)};
\addplot[style=dashed] coordinates {(23.25,0) (23.25,100)};
\end{axis}
\end{tikzpicture}
\caption{Memory Hard Functions: $\%$ cracked by value $v^{\$} \in \{\$4,\$30\}$ adversary against an ideal MHF with running time parameter $\tau$. }
\label{fig:mia:cracked}
\end{figure}

%\paragraph{Discussion} 

%% file: related.tex
%!TEX root=sp-main.tex
\section{Related Work}

\subsubsection{Password Cracking} The issue of offline password cracking has been known for decades~\cite{Morris1979}. Password cracking tools have improved steadily as researchers have explored probabilistic password models ~\cite{Narayanan2005}, Probabilistic Context Free Grammars for passwords~\cite{Weir2009,Kelley2012,Veras2014}, Markov chain models~\cite{Castelluccia2012,Castelluccia2013,Ma2014study,cracking:usenix15} and even neural networks~\cite{Melicher2016}. Attackers may also use public resources (e.g., quotes from the Internet Movie Database or project Gutenberg to crack sentence based passwords~\cite{Yan2000,YBAG04}) as well as `training data' from previous breaches at companies like RockYou or Tianya to improve cracking algorithms. Improved password cracking tools make it all the more crucial to develop secure tools for key-stretching (e.g., data-independent MHFs) to minimize the number of guesses an attacker can try. Allodi has studied the economics of the black market for certain attacks and malware, which may be useful in understanding how password cracking markets may work \cite{russianBlackMarket}.

\subsubsection{Improving Password Strength} Efforts to encourage (or force) users to select stronger passwords have shown limited success~\cite{campbell2011impact,Komanduri2011,Shay2010,Stanton2005,Inglesant2010,Shay2014} and often induce high usability costs~\cite{Adams1999}.  Users can be encouraged to select stronger passwords by providing feedback during the password creation policy (e.g., ~\cite{Komanduri2014,Ur2012,de2014very}) or by providing clear instructions for the user to follow when creating passwords~\cite{Yan2000,BlockiKCD15}. Another extensive line of research explored the use of password composition policies in which a user is required to select a password satisfying certain requirements e.g., contains numbers and/or capital letters~\cite{campbell2011impact,Komanduri2011,Shay2010,Stanton2005,Shay2014,Shay2016}. Password composition policies also introduce a high usability cost~\cite{Inglesant2010,NIST2014,Florencio2014lisa,Adams1999}, and they typically do not increase password strength significantly. In fact, sometimes these policies result in weaker user passwords~\cite{blockiPasswordComposition,Komanduri2011}. Similarly, password strength meters often provide inconsistent feedback~\cite{Ur2012,de2014very} and they often fail to persuade users to select strong passwords. 

Another line of research has focused on helping users to generate and remember passwords. One prominent suggestion is to turn a phrase or a sentence into a password. It has been claimed that these passwords are as strong as random ones~\cite{Yan2000,YBAG04}, and this has been promoted by NIST and by security experts such as Bruce Schneier~\cite{SchneierChoosingSecurePasswords}. However, subsequent research indicates that these suggestions are less secure than previously believed~\cite{Kuo2006,Yang2016ccs}. Another line of research seeks to develop and promote secure and usable strategies for password management when the user needs to create and remember multiple passwords~\cite{blocki2013naturally,Florencio2014,blum2015publishable,Blocki2017HumanComputable}. However, all of these schemes require a motivated user. Bonneau and Schechter~\cite{BonneauS14} and Blocki et al.~\cite{BlockiKCD15} showed that users are capable of memorizing higher entropy secrets (e.g., $56$ bits) by following spaced repetition schedules.

\subsubsection{Other Defenses Against Offline Attacks} If an organization has multiple authentication servers then they could distribute storage and/or computation of the password hashes across multiple servers~\cite{brainard2003new,camenisch2012practical,everspaugh2015pythia,phoenix}. Juels and Rivest~\cite{rivestHoneywords} proposed storing the hashes of fake passwords (honeywords) and using a second auxiliary server to detect authentication attempts with honeywords (alerting the organization that an breach has occurred). The expensive requirement to purchase and maintain extra servers may prevent widespread adoption of these proposals. Even if these defenses were adopted there is still a clear need to use secure key-stretching mechanisms --- an adversary who breaches both servers can still mount an offline attack. Another line of research has sought to include the solution(s) to hard artificial intelligence problems in the password hash so that an offline attacker needs human assistance to verify each password guess~\cite{canetti2006mitigating,blockiGOTCHA,Blocki2016POH}. These solutions increase user workload during authentication e.g., by requiring the user to solve a CAPTCHA puzzle~\cite{canetti2006mitigating,Blocki2016POH}.

\subsubsection{Modeling the Distribution of User Selected Passwords} Malone and Kevin initially explored the feasibility of modeling the distribution of user password choices using Zipf's law~\cite{malone2012investigating}. Wang et al.~\cite{WangZipfLaw14} and Wang and Wang~\cite{WangW16} continued this line of work by providing improved techniques to fit Zipf's law parameters to a dataset. Bonneau~\cite{bonneau2012yahoo} took a different approach: collect and analyze a massive password frequency corpus with permission from Yahoo! The Yahoo! dataset was recently released using a differentially private algorithm~\cite{DPPLists}. We elaborate on Zipf's law and the Yahoo! frequency corpus at length in the body of the paper. 

\subsection{Key-Stretching} Key-stretching was proposed as early as 1979~\cite{Morris1979} with the goal of protecting lower-entropy secrets like passwords against offline attacks by making it economically infeasible for an offline attacker to try millions or billions of guesses. Traditionally key stretching has been performed using hash iteration e.g., PBKDF2~\cite{kaliski2000pkcs} and BCRYPT~\cite{provos1999bcrypt}. However, password hash functions like PBKDF2 and BCRYPT require minimal memory to evaluate and thus passwords protected by these hash functions are highly vulnerable to attackers with customized hardware~\cite{Antminer}. Memory hard functions (MHFs), first explicitly introduced by Percival~\cite{Per09}, are a promising tool for constructing an ideal key-stretching function. MHFs are motivated by the observation that the cost of storing/retrieving items from memory is relatively constant across different computer architectures. At a high level a memory hard function is moderately expensive to compute and most of the costs  associated with computing the function are memory related (e.g., storing/retrieving items from memory). Ideally we want the Area x Time complexity of computing a MHF to scale with $\tau^2$, where $\tau$ denotes the running time on a standard PC. Intuitively, to compute the MHF once the attacker must dedicate $\tau$ blocks of memory for $\tau$ time steps, which ensures that the cost of computing the function is equitable across different computer architectures (memory on an ASIC is still expensive). By contrast, Area x Time complexity to compute BCRYPT or PBKDF2 is simply $\tau$. Recall that we want to increase costs quickly to minimize delay during authentication. If costs scale with $\tau^2$ then we can rapidly drive up costs, and if computation requires memory then an adversary will not be able to significantly reduce guessing costs by constructing an ASIC. Almost all of the entrants to the recent Password Hashing Competition (PHC)~\cite{PHC} claimed some form of memory-hardness.  

\subsubsection{Data (In)dependent Memory Hard Functions} 
There is a type of MHF called a data-independent MHF (iMHF) which is designed to be resistant to side-channel attacks such as cache timing~\cite{Ber05,forler2013catena}. These functions have a data access pattern independent of the input. Multiple attacks have been shown in several iMHFs \cite{BK15,AS15,AB16,AB16b,TCC:BloZho17,CCS:AlwBloHar17,EC:AlwBloPie17}. Data dependent MHFs such as SCRYPT~\cite{Per09} have the previously mentioned side-channel vulnerabilities. Even so, SCRYPT has been found to be optimally memory hard with respect to AT complexity~\cite{ACKKPT16,ACPRT17}. The authors of Argon2~\cite{Argon2}, winner of the password hashing competition~\cite{PHC}, now recommend running in {\em hybrid mode} Argon2id to balance side-channel resistance and resistance to iMHF attacks.

%% file: discussion.tex
% !TEX root =sp-main.tex
\section{Discussion} \label{sec:Discussion}
Our economic analysis decisively shows that traditional key-stretching tools like PBKDF2 and BCRYPT fail to provide adequate protection for user passwords, while memory hard functions do provide meaningful protection against offline attackers. It is time for organizations to upgrade their password hashing algorithms and adopt modern key-stretching such as memory hard functions~\cite{Per09,PHC}. Alternatively, could a creative organization adapt customized Bitcoin mining rigs for use in password authentication? For example, the Antminer S9~\cite{Antminer}, currently available on Amazon for approximately $\$3,000$, is capable of computing SHA256 $14$ trillion times per second. If the organization stored salted and peppered~\cite{manber1996simple,BlockiD16} password hash values $u,s_u, SHA256(pwd_u|s_u|p_u)$  then it could potentially use the Antminer S9, or a similar Bitcoin mining rig, to validate a password by quickly enumerating over a (very) large space of secret pepper values $p$ (briefly, a secret salt value that is not stored which even an honest party must brute force). 

While our analysis demonstrates that the use of memory hard functions can significantly reduce the fraction of cracked passwords, the damage of an offline attack may still be significant. Thus, we recommend that organizations adopt distributed password hashing~\cite{brainard2003new,camenisch2012practical,everspaugh2015pythia,phoenix} whenever feasible so that an attacker who only breaches one authentication server will not be able to mount an offline attack. Furthermore, we recommend that organizations take additional measures to mitigate the affect of an authentication server breach. Solutions might include mechanisms {\em detect} password breaches through the use of honey accounts or honey passwords\cite{rivestHoneywords}, multi-factor authentication and fraud detection/correction algorithms to prevent suspicious/harmful behavior~\cite{NDSS:FJDBG16}. 

%In addition to improved password hashing strategies organizations must find ways to mitigate the affect of an authentication server breach. Solutions might include mechanisms to detect password breaches~\cite{rivestHoneywords}, distributed password hashing~, multi-factor authentication and fraud detection/correction algorithms to prevent suspicious/harmful behavior. 

While solid options for password hashing and key-derivation exist~\cite{Per09,PHC,Argon2,CCS:AlwBloHar17} the reality is that many organizations and developers select suboptimal password hashing functions~\cite{bonneau2011,PasswordStorageUserStudy}. Thus, there is a clear need to provide developers with clear guidance about selecting secure password hash functions. On a positive note recent 2017 NIST guidelines do {\em suggest} the use of memory hard functions. However, NIST guidelines still allows for the user of PBKDF2 with just $10,000$ hash iterations. Based on our analysis we advocate that password hashing standards should be updated to require the use of memory hard functions for password hashing and disallow the use of non-memory hard functions such as BCRYPT or PBKDF2. It may be expedient for policy makers to audit and/or penalize organizations that fail to follow appropriate standards for password hashing. 

We recommend that users primarily focus on selecting passwords that are strong enough to resist targeted online attacks~\cite{wang2016targeted} as there is a often a vast gap between the required entropy to resist online and offline attacks~\cite{bonneau2015passwords}. Extra user effort to memorize a high entropy password might be completely wasted if an organization adopts poor password hashing algorithms like SHA1, MD5~\cite{AMWeak} or the identity function~\cite{bonneau2011}. This effort would likely be more productively spent on trying to reduce password reuse~\cite{Florencio2014}.

\section{Acknowledgments}
We would like the thank the reviewers for their insightful comments. We would also like to thank Ding Wang for sharing code for computing Zipf fittings. The work was supported by the National Science Foundation under NSF Awards \#1649515 and \#1704587. Ben Harsha was partially supported by a Intel Graduate Research Assistantship through CERIAS at Purdue. The opinions expressed in this paper are those of the authors and do not necessarily reflect those of the National Science Foundation or Intel.

%% file: missingProofs.tex
% !TEX root=sp-main.tex

\section{Missing Proofs} \label{apdx:MissingProofs}

\begin{remindertheorem}{Theorem \ref{thm:ModelIndependentLB1}}
\thmLowerBound
\end{remindertheorem}

\begin{proofof}{Theorem \ref{thm:ModelIndependentLB1}}

We observe that a user password $pwd^t = pwd_i$ will be certainly cracked if $V\times \Pr[pwd] = V p_i \geq k$ since the marginal cost of including an extra guess $pwd$ in the dictionary is at most $k$. Thus, the adversary will compromise at least $\sum_{i: p_i \geq k/V} f_i$ accounts. The problem with this lower bound is that we need to know  $p_i= \Pr[pwd_i]$ for each password $pwd_i$ to compute it. However, the  values $p_i$ are unknown if we do not make assumptions about the shape of the password distribution. However, we can lower bound this quantity. In particular, we say that the estimate $\hat{p}_i = f_i/N$ is a {\em $(j,L)$-bad overestimate} if $p_i < \frac{1}{NL}$, but $f_i \geq j$. Let $B_i$ be an indicator random variable for the that $\hat{p}_i$ is $(j,L)$-bad. Then the sum $\sum_{i} f_i \times B_i$ computes the total number of users whose password got a $(j,L)$-bad overestimate. The proof now follows from Claims \ref{claim:ModelIndependentLB1}, \ref{claim:probabilityBadOverestimate} and \ref{claim:numberBadOverestimates}. Claim \ref{claim:ModelIndependentLB1} lower bounds the fraction of cracked passwords in terms of the events $B_i$.
\begin{claim} \label{claim:ModelIndependentLB1}
If $\frac{V}{k} \geq NL$ then the number of user passwords in our dataset that a rational adversary cracks is at least 
\[ \sum_{i: f_i\geq j} f_i - \sum_{i} f_i \times B_i \  . \]
\end{claim} 
\begin{proof}
Suppose that a user selects a password $pwd_i$ with $p_i \geq \frac{1}{NL}$. Since $Vp_i > k \geq \max_{t} MC(t)$ the marginal reward of guessing $pwd_i$ always exceeds the marginal cost. Thus, a rational attacker {\em must} eventually guess $pwd_i$. However, if $p_i < \frac{1}{NL}$ then either $f_i < j$ or $\hat{p}_i$ is a $(j,L)$-bad overestimate of $p_i$. Let $S$ denote the set of users who picked a password $i$ such that $f_i \geq j$ and let $T \subseteq S$ denote the set of users whose password got a $(j,L)$-bad overestimate. Any user in the set $S\backslash T$ will be compromised eventually. Thus, at least $\left|S\backslash T\right| = \sum_{i: f_i\geq j} f_i- \sum_{i} f_i \times B_i$ since  $|T|=\sum_{i} f_i \times B_i$ and  $|S|=\sum_{i: f_i\geq j} f_i$. 
\end{proof}
Claim \ref{claim:probabilityBadOverestimate} bounds the probability of the event $B_{pwd}$ --- that is the probability that we observe password $pwd_i$, with $p_i < \frac{1}{NL}$, at least $f_i \geq j$ times conditioned on the event that we observe $pwd_i$ at least once.  

\begin{claim} \label{claim:probabilityBadOverestimate}
$\Pr[B_i|f_i\ge 1] \leq \frac{1}{(j-1)!L^{j-1}}$
\end{claim}
\begin{proof}
We first observe that 
\[\Pr[B_i|f_i\ge 1] \leq {N \choose j-1} p_i^{j-1} \ . \]
Recall that for an event which is $(j,L)$-bad, we have that by definition, $p_i<\frac{1}{NL}$. 
Thus
\begin{align*}
{N \choose j-1} p_i^{j-1} &= \frac{N!}{(j-1)!(N-j+1)!}p_i^{j-1} \\
&\le\frac{N^{j-1}}{(j-1)!}p_i^{j-1}\le\frac{1}{(j-1)!L^{j-1}}.
\end{align*}
\end{proof}

Finally, Claim \ref{claim:numberBadOverestimates} shows that we cannot have too many bad overestimates. 

\begin{claim} \label{claim:numberBadOverestimates}
$\mathbb{E}\left[\sum_{i} f_i \times B_i\right] \leq \frac{N}{(j-1)!L^{j-1}}$
\end{claim}
\begin{proof}
Consider drawing $N$ samples $x_1,\ldots,x_N$ from our password distribution. Let $Y_i$ be and indicator random variable for the event that the password $x_i$ was sampled at least $j-1$ additional times even though $\Pr[x_i] \leq \frac{1}{NL}$.   Observe that
\[\sum_{i} f_i \times B_i = \sum_{i \leq N} Y_i \ . \]
By Claim \ref{claim:probabilityBadOverestimate} we have $\Pr[Y_i=1] \leq \frac{1}{(j-1)!L^{j-1}}$ for all $i \leq N$.
Thus, 

\begin{align*}
\mathbb{E}\left[\sum_{i} f_i \times B_i\right]  &= \mathbb{E}\left[\sum_{i \leq N} Y_i \right]\\
& \leq  N \max_{i} \Pr[Y_i=1] \\
&\le\frac{N}{(j-1)!L^{j-1}} \ .
\end{align*}

%Thus the statement follows immediately from Claim \ref{claim:probabilityBadOverestimate} by substituting $\frac{1}{(j-1)!L^{j-1}}$ for $\mathbb{E}\left[ B_i\right]$ in the above sum.
\end{proof}

\end{proofof}

\begin{remindertheorem}{Theorem \ref{thm:ModelIndependentUB1}}
\themModelIndepUpperBound
\end{remindertheorem}

\begin{proofof}{Theorem \ref{thm:ModelIndependentUB1}}
Given $N$ independent samples $pwd^1,\ldots, pwd^N \leftarrow \mathcal{X}$ we use $\textsf{Pop}_t = \{pwd_1,\ldots,pwd_t\}$ to denote the $t$ most common passwords from these samples and let 
 $X_i$ be an indicator variable for the event $ pwd^i \in \textsf{Pop}_t$. 

\begin{claim}\label{claim:first}
\[\sum_{i=1}^N X_i\le \sum_{i=1}^t f_i.\]
\end{claim}
\begin{proof}
Let $pwd_1,\ldots, pwd_t, \ldots$ denote the list of observed passwords ordered by observed frequency. Let $i_1 > \ldots >i_t$ be given such that $\textsf{Pop}_t = \{pwd_{i_1},\ldots, pwd_{i_t}\}$. Now we have \[\sum_{j=1}^N X_j = \sum_{j=1}^t f_{i_j} \leq \sum_{j=1}^t f_j \ . \]

\end{proof}

\jnote{This claim will probably move to the appendix, but we should quantify the probability in terms of $\epsilon$ and $N$ in the claim itself. Note that we will typically have $\sum_{i=1}^t p_i \geq 0.1$ so we could set $\epsilon = 0.001$ and the result still holds except with probability $0.03$. We should say something like this in the text...}. 
\begin{claim}\label{claim:chernoff}
We have 
\[\sum_{i=1}^{t} p_i\le\frac{(1+\eps)}{N}\sum_{i=1}^N X_i \]
except with probability 
\[\PPr{\sum_{i=1}^N X_i\le\frac{N}{1+\eps}\sum_{i=1}^t p_i}\le\exp\left(-\frac{\eps^2N\sum_{i=1}^t p_i}{2(1+\eps)^2}\right).\]
\end{claim}
\begin{proof}
Since $\PPr{X_i=1}=\sum_{i=1}^{t} p_i$, then $\EEx{\sum_{i=1}^N X_i}=N\sum_{i=1}^{t} p_i$. 
Then applying Chernoff bounds,
\[\PPr{\sum_{i=1}^N X_i\le\frac{N}{1+\eps}\sum_{i=1}^t p_i}\le\exp\left(-\frac{\eps^2N\sum_{i=1}^t p_i}{2(1+\eps)^2}\right).\]
\end{proof}

\begin{claim}
With high probability,
\[MC(t)\ge\left(1-\frac{1+\eps}{N}\sum_{i=1}^t f_i\right)k.\]
\end{claim}
\begin{proof}
By Claims~\ref{claim:first} and \ref{claim:chernoff},
\[\sum_{i=1}^t p_i\le\frac{1+\eps}{N}\sum_{i=1}^t f_i.\]
The proof follows from the observation that $MC(t)=\left(1-\sum_{i=1}^{t-1}p_i\right)k$ and $p_t\ge 0$.
\end{proof}

Now, we define $\hat{p}_i = f_i / N$ as a $(j, L)$-\textit{bad underestimate} if $p_i > \frac{1}{NL}$, but $f_i\le j$. 
Then define $C_i$ as the indicator variable for the event that $\hat{p}_i$ is a $(j,L)$-bad underestimate and $f_i\ge 1$.

\begin{claim} \label{claim:ModelIndependentUB1}
If $\frac{V}{k} \leq NL\left(1-\frac{1+\eps}{N}\sum_{i=1}^t f_i\right)$ then the number of user passwords in our dataset that a rational adversary cracks is at most 
\[ \sum_{i: f_i\geq j} f_i + \sum_{i} f_i \times C_i \  . \]
\end{claim} 
\begin{proof}
Suppose that a user selects a password $pwd_i$ with $p_i \leq \frac{1}{NL}$. 
Since $Vp_i < \left(1-\frac{1+\eps}{N}\sum_{i=1}^t f_i\right)k \leq MC(t)$ the marginal reward of guessing $pwd_i$ never exceeds the marginal cost. 
Thus, a rational attacker {\em never} chooses to guess $pwd_i$. 
If $p_i > \frac{1}{NL}$ then either $f_i > j$ or $\hat{p}_i$ is a $(j,L)$-bad underestimate of $p_i$. 
Let $S$ denote the set of users who picked a password $i$ such that $f_i > j$ and let $T \subseteq S$ denote the set of users whose password got a $(j,L)$-bad underestimate. \jnote{$T$ is not a subset of $S$ though the derivation below looks correct to me}
Only the users in the set $S\cup T$ may be compromised eventually.
Thus, at most $\left|S\cup T\right| \leq \sum_{i: f_i > j} f_i + \sum_{i} f_i \times C_i$ since  $|T|=\sum_{i: f_i \le j} f_i \times C_i$ and  $|S|=\sum_{i: f_i> j} f_i$. 
\end{proof}
Then the following immediately holds, noting that there can be at most $NL$ passwords which are $(j,L)$-bad underestimates:
\begin{corollary}
If $\frac{V}{k} \leq NL\left(1-\frac{1+\eps}{N}\sum_{i=1}^t f_i\right)$ then the number of user passwords in our dataset that a rational adversary cracks is at most 
\[ \sum_{i: f_i\geq j} f_i + \sum_{i}^{i+NL} f_i\  . \]
\end{corollary}

\begin{claim} \label{claim:probabilityBadUnderestimate}
\[\Pr\left[C_i~\vline~ f_i \geq 1 \right]\le\sum_{\ell=0}^{j-1}\binom{N-1}{\ell} \left( \frac{1}{NL}\right)^\ell\left(\frac{NL-1}{NL}\right)^{N-\ell-1}\]
\end{claim}
\begin{proof}
Recall that for $C_i=1$, we require $p_i>\frac{1}{NL} \geq \frac{j}{N}$ but $f_i\le j$. 
Then for $j<N/2$,
\begin{align*}
\Pr\left[C_i~\vline~ f_i \geq 1 \right]&=\sum_{\ell=0}^{j-1}\binom{N-1}{\ell}p_i^\ell(1-p_i)^{N-\ell-1}\\
&\le \sum_{\ell=0}^{j-1}\binom{N-1}{\ell} \left( \frac{1}{NL}\right)^\ell\left(\frac{NL-1}{NL}\right)^{N-\ell-1}  \\
\end{align*}
%\jnote{To derive above note that each term $\binom{N-1}{\ell}p_i^\ell(1-p_i)^{N-\ell-1}$ is maximized when $p_i = \frac{\ell}{N-1}$, but that we have $p_i>\frac{1}{NL}> \frac{\ell}{N-1}$ for each $\ell < j$. Plugging in $\frac{1}{NL}$ thus gives an upper bound for each term.} 
%
%\jnote{We lose way too much in the next steps...I would either omit them or use Sterling to continue simplification. In particular, we definitely need the term $\left(\frac{NL-1}{NL}\right)^{N-k-1}$. if $L=0.5$ and $N=7\times 10^7$ then $\left(1-\frac{1}{NL}\right)^{N-100} \approx 0.135$ and if $L= 1/7$ then this term is $\approx 0.0009$. It may be ok not to simplify further since readers will be most interested in the final numerical result. }
%\begin{align*}
%&\le\sum_{\ell=1}^{j-1}\binom{N}{\ell}\left(\frac{1}{NL}\right)^\ell\\
%&\le\frac{j-1}{L}
%\end{align*}
\end{proof}

\begin{claim} \label{claim:numberBadUnderestimates}
$\sum_{i: 0 < f_i \leq j} f_i \times \mathbb{E}\left[ C_i~ \vline~ f_i \geq 1\right] $.
\end{claim}
\begin{proof}
Follows immediately from Claim \ref{claim:probabilityBadUnderestimate} by substituting into $\mathbb{E}\left[C_i~ \vline~ f_i \geq 1\right]$ in the above sum.
\end{proof}

\end{proofof}

\section{Extra Figures}

\begin{figure}
\centering
\includegraphics[scale=0.31]{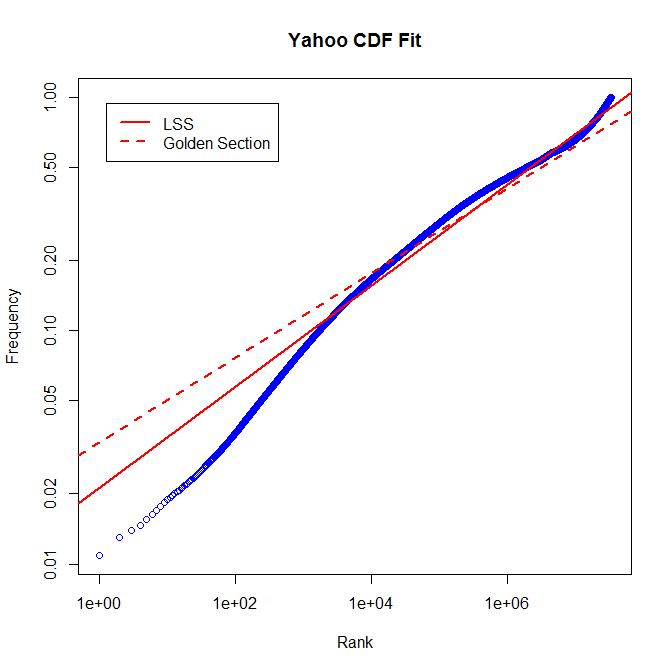}
\caption{Yahoo! CDF-Zipf Fittings}
\label{fig:BothFit}
\end{figure}

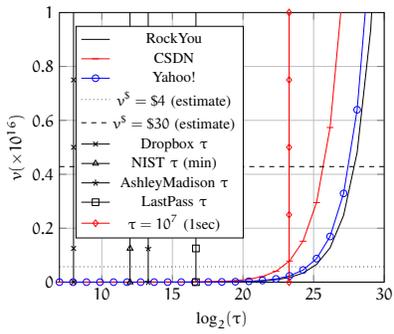
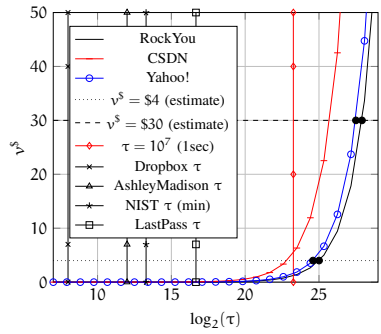
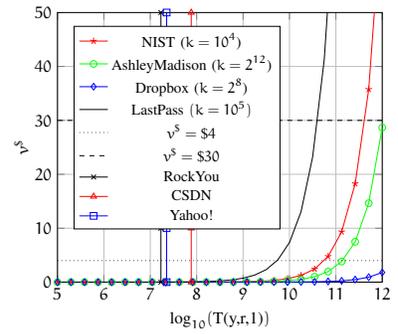
\begin{figure*}[!h]
\centering
\subfigure[$v/k = T(y,r,1)$ for RockYou, CSDN and Yahoo!]{
\centering
\begin{tikzpicture}[scale=0.63] 
\begin{axis}[title style={align=center},
   xlabel={$\log_2(\tau)$},
   ylabel={$v(\times10^{16})$},
	 xmin={7},
	 xmax={30},
	 ymin={0},
	 ymax={1},
   ylabel shift = -3pt,
   grid=major,
	 cycle list = {{red, mark=circle},  {orange, mark=none}, {purple, mark=none}},
   legend style = {font=\small, at={(.05,.95)}, anchor=north west},
	 legend entries = {RockYou, CSDN, Yahoo!, $v^\$=\$4$ (estimate),$v^\$=\$30$ (estimate), Dropbox $\tau$, NIST $\tau$ (min), AshleyMadison $\tau$, LastPass $\tau$, $\tau=10^7$ (1sec)}
  ]

\addplot[color=black, mark=r, domain=7:30]{1.69657*10^7*(2^x)/10^16};
%CSDN
\addplot[color=red, mark=-, domain=7:30]{7.63300*10^7*(2^x)/10^16};
%enter yahoo T(y,r) value here:
\addplot[color=blue, mark=o, domain=7:30]{2.25435*10^7*(2^x)/10^16};
\addplot[style=dotted, domain=7:30]{0.057143};
\addplot[style=dashed, domain=7:30]{0.057143*30.0/4.0};
\addplot[mark=x] coordinates{ (8, 0)  (8,0.125) (8,0.5) (8,0.75) (8, 50)};
\addplot[mark=triangle] coordinates{ (12, 0)(12,0.125)  (12,0.5) (12,0.75) (12, 50)};
\addplot[mark=star] coordinates {(13.2877,0) (13.2877,0.125) (13.2877,0.5) (13.2877,0.75) (13.2877,50)};
\addplot[mark=square] coordinates{ (16.66667, 0) (16.66667,0.125) (16.66667,0.5) (16.66667,0.75) (16.66667, 50)};
\addplot[color=red,mark=diamond] coordinates {(23.2635,0) (23.2635,0.25) (23.2635,0.50) (23.2635,0.750) (23.2635,1)};
\end{axis} 
\end{tikzpicture}
%\caption{$v/k = T(y,r)$ for RockYou and Yahoo! plus $\tau=k$ for Dropbox, AshleyMadison and LastPass.}
\label{fig:0.8:1}}
\subfigure[$v^{\$}$ vs. $\tau$ for $v = k \times T(y,r,1)$. ]{
\begin{tikzpicture}[scale=0.63] 
\begin{axis}[title style={align=center},
  xlabel={$\log_2(\tau)$},
  ylabel={$v^\$$},
	xmin={7},
	xmax={29},
	ymin={0},
	ymax={50},
  ylabel shift = -3pt,
  grid=major,
  legend style = {font=\small, at={(.05,.95)}, anchor=north west},
	legend entries = {RockYou, CSDN, Yahoo!, $v^{\$}=\$4$ (estimate), $v^{\$}=\$30$ (estimate),$\tau=10^7$ (1sec), Dropbox $\tau$, AshleyMadison $\tau$, NIST $\tau$ (min), LastPass $\tau$}
  ] 
\addplot[color=black, domain=7:29]{(1.69657*10^7)*(7*10^-15)*(2^x)};
%CSDN
\addplot[color=red, mark=-,domain=7:29]{(7.63300*10^7)*(7*10^-15)*(2^x)};
%Enter yahoo T(y,r) value here:
\addplot[color=blue, mark=o, domain=7:29]{(2.25435*10^7)*(7*10^-15)*(2^x)};
\addplot[color=black, style=dotted, domain=7:29]{4};
\addplot[color=black, style=dashed, domain=7:29]{30};
\addplot[color=red,mark=diamond] coordinates {(23.2635,0) (23.2635,20) (23.2635,40) (23.2635,50)};
\addplot[mark=x] coordinates{ (8, 0)  (8,7) (8,40) (8, 50)};
\addplot[mark=triangle] coordinates{ (12, 0)(12,7)  (12,40) (12, 50)};
\addplot[mark=star] coordinates {(13.2877,0) (13.2877,7) (13.2877,40) (13.2877,50)};
\addplot[mark=square] coordinates{ (16.66667, 0) (16.66667,7) (16.66667,40)  (16.66667, 50)};
%rockyou coordinates
\addplot[mark=*] coordinates {(25.00545,4)};
\addplot[mark=*] coordinates {(27.91234,30)};
%tianya coordinates
%\addplot[mark=*] coordinates {(24.57985,4)};
%\addplot[mark=*] coordinates {(27.48674,30)};
%yahoo coordinates
%solve((2.25435*10^7)*(7*10^-15)*(2^x)=30)
\addplot[mark=*] coordinates {(24.5954,4)};
\addplot[mark=*] coordinates {(27.5022,30)};
\end{axis} 
\end{tikzpicture}
%\caption{$v^{\$}$ vs. $\tau$ for $v = k \times T(y,r)$. }
\label{fig:0.8:2}}
\subfigure[$v^{\$}$ versus $T(y,r,1)$ when $v = k\times T(y,r,1)$, at fixed values of $k$]{
\begin{tikzpicture}[scale=0.63] 
\begin{axis}[title style={align=center},
   xlabel={$\log_{10}($T(y,r,1)$)$},
   ylabel={$v^{\$}$},
	 xmin={5},
	 xmax={12},
	 ymin={0},
	 ymax={50},
	 %extra x ticks = {7.04034, 6.741211, 5.728109, 7.35503},
	 %extra x tick labels = {},
	 %extra x tick style = {major grid style=red, tick align=outside, tick style=red},
   ylabel shift = -3pt,
   grid=major,
	 cycle list = {{red, mark=none},  {orange, mark=none}, {yellow, mark=none}, {purple, mark=none}},
	 legend style = {font=\small, at={(.05,.95)}, anchor=north west},
	 legend entries = {NIST ($k=10^4$), AshleyMadison ($k=2^{12}$), Dropbox ($k=2^8$), LastPass $(k=10^5)$, $v^{\$}=\$4$, $v^{\$}=\$30$, RockYou, CSDN, Yahoo!}
  ] 
\addplot[color=red,mark=star, domain=5:12]{(10^x)*(7*10^-15)*(10^4)};
\addplot[color=green,mark=o, domain=5:12]{(10^x)*(7*10^-15)*(2^12)};
\addplot[color=blue, mark=diamond, domain=5:12]{(10^x)*(7*10^-15)*(2^8)};
\addplot[color=black, domain=5:11]{(10^x)*(7*10^-15)*(5000+10^5)};
\addplot[color=black, style=dotted, domain=5:12]{4};
\addplot[color=black, style=dashed, domain=5:12]{30};

%rockyou
\addplot[mark=x] coordinates{ (7.229572, 0) (7.229572, 10) (7.229572, 50)};
%CSDN
\addplot[mark=triangle, color=red] coordinates{ (7.882695, 0) (7.882695, 10)  (7.882695, 50)};
%insert log(T(y,r)) base 10 for Yahoo here
\addplot[mark=square, color=blue] coordinates{ (7.353021, 0) (7.353021, 10) (7.353021, 50)};

\end{axis} 
\end{tikzpicture}
%\caption{$v^{\$}$ versus $T(y,r)$ when $v = k\times T(y,r)$, at fixed values of $k$ selected by AshleyMadison, Dropbox and LastPass. Vertical lines show thresholds $T(y,r)$ for  RockYou and Yahoo! }
\vspace{-0.2cm} 
\label{fig:0.8:3}}
\caption{No Diminishing Returns ($a=1$)}
\end{figure*}